\documentclass[10pt, doublecolumn]{IEEEtran}
\usepackage{epsfig,latexsym}
\usepackage{float}
\usepackage{indentfirst}
\usepackage{amsmath}
\usepackage{bm}
\usepackage{amssymb}
\usepackage{times}
\usepackage{enumitem}

\usepackage{algorithm}
\usepackage[noend]{algpseudocode}
\usepackage{subfigure}
\usepackage{psfrag}
\usepackage{hyperref}
\usepackage{cite}
\usepackage{lastpage}
\usepackage{fancyhdr}
\usepackage{color} 
 \usepackage{amsthm}
\usepackage{bigints}
\usepackage{array}
\usepackage{booktabs}
\newtheorem{Lemma}{Lemma}
\newtheorem{Corollary}{Corollary}
\newtheorem{lemma}[Lemma]{$\mathbf{Lemma}$}

\newtheorem{corollary}[Corollary]{$\mathbf{Corollary}$}

\newcounter{problem}
\newcounter{save@equation}
\newcounter{save@problem}
\makeatletter
\newenvironment{problem}
{\setcounter{problem}{\value{save@problem}}%
  \setcounter{save@equation}{\value{equation}}%
  \let\c@equation\c@problem
  \subequations
}
{\endsubequations
  \setcounter{save@problem}{\value{equation}}%
  \setcounter{equation}{\value{save@equation}}%
}

\begin{document}
\title{  \vspace{-0.5em}\huge{  Analytical Optimization for Antenna Placement in Pinching-Antenna Systems }}

\author{ Zhiguo Ding, \IEEEmembership{Fellow, IEEE}, and H. Vincent Poor, \IEEEmembership{Life Fellow, IEEE}   \thanks{ 
  
\vspace{-1em}

Z. Ding is with the University
of Manchester, Manchester, M1 9BB, UK, and Khalifa University, Abu Dhabi, UAE.    
H. V. Poor is with the  Department of Electrical and Computer Engineering, Princeton University, Princeton, NJ 08544, USA.
 

  }\vspace{-2em}}
 \maketitle

\begin{abstract}
As the main issue in pinching-antenna system design, antenna location optimization is key to realizing channel reconfigurability and system flexibility. Most existing works in this area adopt sophisticated optimization and learning tools to identify the optimal antenna locations in a numerical manner, where insightful understandings of the pinching antenna placement are still missing. Motivated by this research gap, this paper aims to carry out analytical optimization for pinching antenna placement, where closed-form solutions for the optimal antenna locations are obtained to reveal the impact of antenna placement on the system performance. In particular, for the user-fairness-oriented orthogonal multiple access (OMA) based transmission, analytical results are obtained to reveal that the pinching antenna needs to be activated at the place that would be beneficial to all served users; however, the users' distances to the waveguide have no impact on the location selection. For the greedy-allocation-based OMA transmission, an asymptotic study based on a high signal-to-noise ratio approximation is carried out to show that the optimal antenna location is in close proximity to the user who is nearest to the waveguide.  For non-orthogonal multiple access (NOMA) based transmission, even with a user-fairness-oriented objective, the obtained analytical results show that the optimal antenna location is not the position that can benefit all users, but rather is near the user positioned closest to the waveguide. 
\end{abstract}\vspace{-0.2em}

\vspace{-0.5em}

\section{Introduction}
The detrimental effects of wireless channels constitute the main difficulty in designing reliable wireless systems \cite{Rappaport}. For example, in a conventional wireless system, the base station is deployed at a fixed location, whereas its serving users are randomly distributed. As a result, the distances between the base station and the users can be substantial, leading to significant large-scale path loss and an increased likelihood of line-of-sight (LoS) blockage. Various communication techniques have been developed in order to combat such impairments introduced by wireless channels. For example, multiple-input multiple-input (MIMO) techniques focus on small-scale multi-path fading, and introduce a high number of spatial degrees of freedom, enabling more flexible and efficient signal transmission  \cite{6375940}. Recently developed flexible-antenna techniques, including reconfigurable intelligent surfaces (RISs), fluid-antenna systems, and movable-antenna networks, are able to reconfigure wireless channels,  enhance signal propagation, reduce interference, and improve system capacity \cite{irs1,9326394,9264694,10243545}. However, these existing techniques are similar in the sense that they accept the loss of LoS links, and in the presence of LoS blockage, their focus is mainly to build favorable small-scale fading

Pinching antennas represent a paradigm shift in the design of wireless communication systems, as their focus is to build strong LoS links between transceivers \cite{pinching_antenna2}. As shown by NTT DOCOMO's testbed, a pinching antenna is similar to a leaky-wave antenna and can be created by applying a separate dielectric, e.g., a cloth pinch, to a dielectric waveguide \cite{pinching_antenna1}. Because radio-frequency signals inside the waveguide can be emitted at an arbitrary point on the waveguide, the use of pinching antennas can create new LoS links between the transceivers or make the existing ones stronger. Various works have been carried out recently to characterize the fundamental limits and illustrate the important applications of pinching antennas. For example, \cite{mypa} identified the performance gains of pinching antennas over conventional fixed-location antennas in scenarios with different numbers of waveguides, antennas, and users. The design of pinching-antenna assisted uplink transmission was considered in \cite{Tegosuppa}, and an outage probability based performance analysis was carried out for pinching-antenna systems in \cite{georgex33}. Practical implementation issues of pinching-antenna systems, including antenna activation and channel estimation, have been investigated in \cite{Kaidiactivation, guispawc}. The potential to use pinching beamforming to realize analog and digital beamforming has been illustrated in \cite{Chongjunpa1x}, and the impact of pinching beamforming on the design of multiple access techniques has been characterized in \cite{yuanwei1xzd}. The impact of LoS blockage on the design of pinching-antenna systems has been analyzed in \cite{myblockage}, where it is shown that LoS blockage brings opportunities to suppress co-channel interference in multi-user pinching-antenna systems.  The applications of pinching antennas to improve secrecy communications and facilitate federated learning have been investigated in \cite{yuanweisecu,250410656}, and  \cite{xiamingco}.  The use of the flexibility and reconfigurability of the pinching-antenna array for improving the sensing resolution and the communication throughput simultaneously has been exploited in \cite{yarupass,yuanwei53x,robertsac}. Furthermore, the capability of pinching antennas to create strong LoS links has also been used to support wireless power transfer as shown in \cite{250606754,250512403}

    \begin{figure}[t]\centering \vspace{-0em}
    \epsfig{file=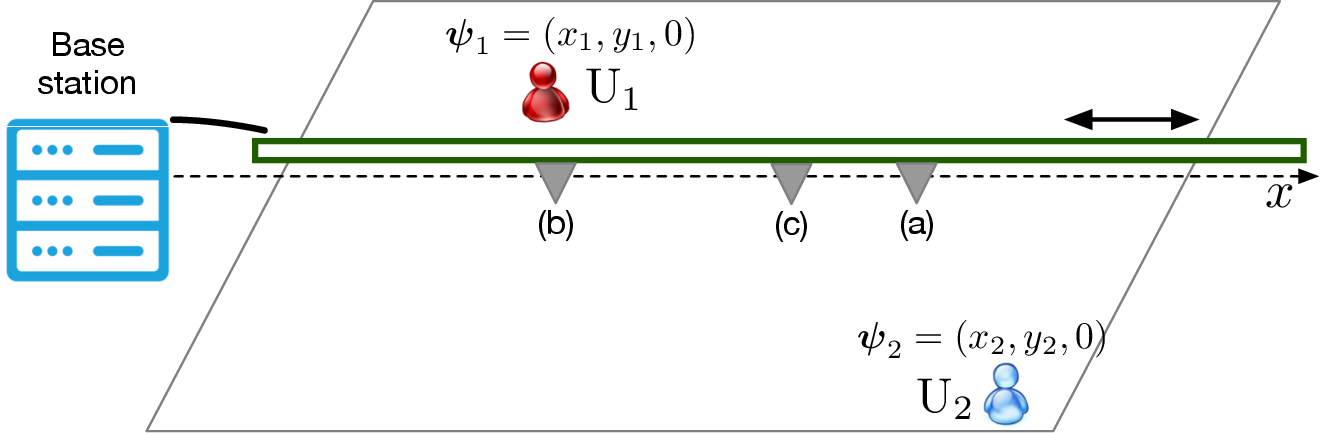, width=0.45\textwidth, clip=}\vspace{-0.5em}
\caption{{An illustration for antenna placement/activation in pinching-antenna systems, with ${\rm U}_m$'s coordinate is denoted by $(x_m,y_m,0)$. Position (a) ensures that the distances from the two users to the antenna are identical, poisition (b) is close to ${\rm U}_1$ which is near the waveguide, and position (c) is midway between the two users on the waveguide, i.e., $\frac{x_1+x_2}{2}$.     }
  \vspace{-1em}    }\label{fig1}   \vspace{-1.2em} 
\end{figure}

Antenna location optimization is the main issue in pinching-antenna system design, as it is key to offering superior channel reconfigurability and system flexibility. For the aforementioned existing works, various optimization or learning based tools have been developed to identify the optimal antenna locations in a numerical manner, where insightful understandings of the pinching antenna placement are still missing. We note that antenna placement was studied for the single-user case in \cite{yanqingpa}, but how the obtained result can be extended to the multi-user scenario is still unknown. Take the two-user case illustrated in Fig. \ref{fig1} as an example. There is still no answer to the simple question of where the pinching antenna should be placed, at position (a) which ensures that all users have the same distance to the antenna, position (b) which is close to a user near the waveguide, or position (c) which is midway between the users on the waveguide?  To bridge this research gap, this paper aims to develop insightful analytical results for the optimal antenna locations in order to reveal the impact of antenna placement on the system performance. The contributions of this paper are listed as follows:
\begin{itemize}
\item For user-fairness-oriented orthogonal multiple access (OMA) based transmission, a resource allocation optimization problem is first formulated, and closed-form analytical expressions are obtained for the optimal power allocation and antenna location. In particular, the obtained closed-form optimal solution reveals that the pinching antenna needs to be activated at a place that would be beneficial to all served users. However, a surprising result is that the users' distances to the waveguide have no impact on the location selection. Take the case shown in Fig. \ref{fig1} as an example, where position (c), instead of position (a), would be used. 

\item For the greedy-allocation-based OMA transmission, the corresponding resource allocation problem is shown to be much more challenging than that of the user-fairness-oriented ones, due to the strong coupling effect between the optimization variables. An asymptotic study based on a high signal-to-noise ratio (SNR) approximation is carried out, where insightful analytical results are obtained to show that the distance between a user and the waveguide is critical to the antenna placement, which is different from the conclusion made in the case with the user-fairness-oriented objectives. In particular, the optimal antenna location needs to be close to the user who is near the waveguide, e.g., position (b) in Fig. \ref{fig1}.  

\item Antenna placement optimization for non-orthogonal multiple access (NOMA) based transmission is formulated and shown to be much more difficult than the OMA ones. The reason is that the design of successive interference cancellation (SIC), the key component of NOMA systems, is affected by the antenna location. Take power-domain NOMA as an example, where a user with strong channel conditions carries out SIC \cite{NOMAPIMRC,Nomading}. However, a change in the location of the pinching antenna can cause a user to degrade from a strong user to a weak user. A low-complexity optimal solution is obtained by using the feature of SIC, and the obtained analytical solution reveals that even with a user-fairness-oriented objective, the optimal location of the pinching antenna is not the position that can benefit all users, e.g., points (a) or (c), but close to the user that is near the waveguide, e.g., position (b) in Fig. \ref{fig1}. 
\end{itemize}

\section{User-Fairness-Oriented OMA Transmission}\label{section 2}
Consider a multi-user downlink communication scenario, where a base station is equipped with a pinching-antenna system and serves $M$ single-antenna users, denoted by ${\rm U}_m$. The OMA-based transmission strategy is focused on in this section, where time division multiple access (TDMA) is used as an example of OMA, i.e., the $M$ users are served separately in $M$ time slots. 

In this paper, we assume that the base station is equipped with a single waveguide, on which a single pinching antenna is activated. An optimal design of this pinching-antenna system is to dynamically change the location of the pinching antenna from one time slot to another, which can lead to high system complexity. A low-complexity approach is to fix the location of the pinching antenna over a moderately prolonged duration, e.g., $M$ time slots, which is assumed throughout this paper. We note that the analytical solutions obtained in this paper are also applicable to orthogonal frequency-division multiple access (OFDMA) systems, where multiple users are served simultaneously with a fixed-location pinching antenna. 

Denote the location of the activated pinching antenna by ${\boldsymbol \psi}_{\rm OMA}$, where  ${\rm U}_m$'s location   is denoted by ${\boldsymbol \psi}_m=(x_m,y_m,0)$. Similar to \cite{mypa}, the users are assumed to be uniformly deployed in a rectangular service area with its two sides being denoted by $D_{\rm W}$ and $D_{\rm L}$, respectively. This means that the location of a user is confined as follows: $-\frac{D_{\rm L}}{2}\leq x_m\leq \frac{D_{\rm L}}{2}$ and $-\frac{D_{\rm W}}{2}\leq y_m\leq \frac{D_{\rm W}}{2}$. The waveguide is placed at the center of the service area with its height denoted by $d$, which means that the coordinate of the pinching antenna is given by ${\boldsymbol \psi}_{\rm OMA}=(x,0,d)$ and $-\frac{D_{\rm L}}{2}\leq x\leq \frac{D_{\rm L}}{2}$. Therefore, 
${\rm U}_m$'s achievable data rate in OMA is given by \cite{mypa}
  \begin{align}\label{models}
  R_m^{\rm OMA} =& \frac{1}{M}\log\left(
  1+\frac{ \eta P_m}{\sigma^2  \left| {\boldsymbol \psi} _m - {\boldsymbol \psi}_{\rm OMA}\right|^2}
  \right)  ,
  \end{align}     
where  $\eta = \frac{c^2}{16\pi^2 f_c^2 }$, $c$ denotes the speed of light, the carrier frequency is denoted by $f_c$, $\left| {\boldsymbol \psi} _m - {\boldsymbol \psi}_{\rm OMA}\right|$ denotes the distance between the pinching antenna and ${\rm U}_m$, $P_m$ denotes the transmit power for  ${\rm U}_m$'s signal, and $\sigma^2$ denotes the noise power. The reason for having the factor $\frac{1}{M}$ is due to the use of TDMA. While the system model is presented for the downlink scenario, the obtained results can be applied to the uplink scenario in a straightforward manner. 
  \vspace{-1.5em}
   \subsection{Maximizing the Worst User's Data Rate}

One user-fairness-oriented resource allocation problem can be formulated as follows: 
    \begin{problem}\label{pb:1} 
  \begin{alignat}{2}
\underset{P_{m}\geq0 , x  }{\rm{max}}  &\quad   \min\left\{R_1^{\rm OMA}, \cdots, R_M^{\rm OMA}\right\} \label{1tst:1}
\\ s.t. &\quad   \sum^{M}_{m=1}  P_{m}\leq P  \label{1tst:2}  
\\  &\quad     -\frac{D_{\rm L}}{2}\leq x \leq \frac{D_{\rm L}}{2} \label{1tst:3}  .
  \end{alignat}
\end{problem}  
Unlike resource allocation for conventional systems with fixed-location antennas, problem \eqref{pb:1} involves not only power allocation but also the antenna location. 
The following lemma provides a precise estimation of the optimal antenna location as well as a closed-form expression of optimal power allocation. 

\begin{lemma}\label{lemma1}
For the user-fairness-oriented resource allocation problem considered in \eqref{pb:1}, the optimal solution of the antenna location is given by
\begin{align}
  x^* =\frac{1}{M}  \sum^{M}_{m=1}x_m ,
\end{align}
and the optimal transmit power solutions are given  by
\begin{align}
P_m^* = \frac{ \tau_{mx^*}   }{  \sum^{M}_{i=1}   \tau_{ix^*} } P,
\end{align}
where $\tau_{mx^*}= (x^*-x_m)^2+y_m^2+d^2$, $1\leq m \leq M$.
\end{lemma}

\begin{proof}
See Appendix \ref{proof1}.

   \end{proof}

   {\it Remark 1:}    Lemma \ref{lemma1} provides precise guidance on the optimal location of the pinching antenna for serving the $M$ user, as well as the optimal power allocation. Recall that if the pinching antenna can be dynamically activated in the $M$ time slots, the optimal antenna placement is to activate the pinching antenna at the position closest to the serving user. However, if the $M$ users are served simultaneously, Lemma \ref{lemma1} shows that the antenna should be activated at a position that is beneficial to all users.
   
   {\it Remark 2:} A surprising observation from Lemma \ref{lemma1} is that the distances of the users to the waveguide, i.e., $y_m$, have no impact on the selection of the optimal antenna location. This is countintuitive, as explained in the following.  Because problem \eqref{pb:1} essentially ensures that the users have the same data rate, intuitively, the optimal antenna location should be the position ensuring that all users have the same distance to the antenna, i.e.,  $\left| {\boldsymbol \psi} _i - {\boldsymbol \psi}_{\rm OMA}\right|=\left| {\boldsymbol \psi} _j - {\boldsymbol \psi}_{\rm OMA}\right|$, for $1\leq i,j\leq M$. Taking the two-user case shown in Fig. \ref{fig1} as an example, position (a) should be optimal. Lemma \ref{lemma1} shows that this intuition is mistaken. Instead, position (c) in Fig. \ref{fig1} should be used. 
   
     \vspace{-1.5em}
 \subsection{Minimizing Total Transmit Power Consumption}
Another user-fairness-oriented resource allocation problem is given by 
    \begin{problem}\label{pb:6} 
  \begin{alignat}{2}
\underset{   x  }{\rm{min}} &\quad   \sum^{M}_{m=1}  P_{m}\label{6tst:1}
\\ s.t. &\quad    R_m^{\rm OMA}\geq R, \quad 1\leq m \leq M.  \label{6tst:2}  
\\  &\quad   P_m\geq 0, \quad -\frac{D_{\rm L}}{2}\leq x \leq \frac{D_{\rm L}}{2} \label{6tst:3}  ,
  \end{alignat}
\end{problem}  
where $R$ denotes the target data rate. 
 
With some straightforward algebraric manipulation, problem \eqref{pb:6} can be recast as follows: 
   \begin{problem}\label{pb:7} 
  \begin{alignat}{2}
\underset{P_{m} , x  }{\rm{min}} &\quad   \sum^{M}_{m=1}  P_{m}\label{7tst:1}
\\ s.t. &\quad    
 P_m
 \geq \epsilon (x-x_m)^2+\tau_m, \quad 1\leq m \leq M.  \label{7tst:2}  
\\  &\quad   \eqref{6tst:3}. 
  \end{alignat}
\end{problem}  
where $\epsilon =\frac{\sigma^2}{\eta}\left(e^{MR}-1\right)$ and $\tau_m=\epsilon \left(y_m^2+d^2\right)$.

Problem \eqref{pb:7} is a linear programming problem, and its closed-form optimal solution can be obtained as follows. In order to reduce the number of Lagrange multipliers, the constraints of $P_m\geq 0$ are omitted, and the Lagrangian of  problem \eqref{pb:7} is given by 
\begin{align}
L = & \sum^{M}_{m=1}  P_{m}+\sum^{M}_{m=1}\lambda_m\left(\epsilon (x-x_m)^2-P_m+\tau_m\right)\\\nonumber &+ \lambda_{M+1}\left(x - \frac{D_{\rm L}}{2}  \right)+ \lambda_{M+2}\left(- \frac{D_{\rm L}}{2} -x \right),
\end{align}
where $\lambda_{m}$ denotes the Lagrange multiplier.
The corresponding Karush-Kuhn-Tucker (KKT) conditions are given by \cite{Boyd}\footnote{Due to space limitations, the primal feasibility conditions are omitted here, but it is straightforward to verify that the obtained solution satisfies all the primal feasibility conditions.  }
\begin{align}\label{kkt2}
\left\{\begin{array}{l}  
1- \lambda_m=0, 1\leq m \leq M\\
2 \sum^{M}_{m=1}\lambda_m\epsilon (x-x_m)+\left(\lambda_{M+1} -\lambda_{M+1}\right)=0 \\
\lambda_m\left(\epsilon (x-x_m)^2-P_m+\tau_m\right) =0 , 1\leq m \leq M\\
\lambda_{M+1}\left(x - \frac{D_{\rm L}}{2}  \right)=0\\ \lambda_{M+2}\left(- \frac{D_{\rm L}}{2} -x \right)=0
 \end{array}\right..
 \end{align}
 Following the steps in the proof for Lemma \ref{lemma1}, it can be established that $\lambda_{M+1}=\lambda_{M+2}=0$. Furthermore, the use of \eqref{kkt2} leads to the conclusion that   $\lambda_m=1$, $1\leq m\leq M$, which yields the following corollary. 
 \begin{corollary}\label{corollary1}
For problem \eqref{pb:6}, the optimal antenna location is given by 
\begin{align}
  x^* =\frac{1}{M}  \sum^{M}_{m=1}x_m ,
\end{align}
and the optimal power control solutions are given by 
\begin{align}
P_m^* =  \epsilon (x^*-x_m)^2 +\tau_m.
\end{align}
\end{corollary}
 It is interesting to point out that the optimal antenna location solutions to problems \eqref{pb:1} and \eqref{pb:6} are identical.  

The closed-form optimal solution of problem \eqref{pb:6} facilitates insightful performance analysis. For example, an analytical comparison between the transmit powers required by pinching and conventional antennas is carried out in the following. In particular,  by using the closed-form expressions of $x^*$ and $P_m^*$, the overall transmit power required by pinching antennas is given by
\begin{align}
P^{\rm pin} = \sum^{M}_{m=1}\left[ \epsilon \left(\frac{1}{M}  \sum^{M}_{i=1}x_i -x_m\right)^2 +\tau_m\right].
\end{align}
whereas the total tranmit power required by conventional antennas is given by
\begin{align}
P^{\rm conv} = \sum^{M}_{m=1}\left[ \epsilon x_m^2 +\tau_m\right].
\end{align}
Therefore, the difference between the transmit powers, denoted by $\Delta^P = P^{\rm conv}  - P^{\rm pin} $, is given by
\begin{align}
\Delta^P 
=&\epsilon \sum^{M}_{m=1}\left[  x_m^2 -\left(\frac{1}{M}  \sum^{M}_{i=1}x_i -x_m\right)^2 \right] .
\end{align} 
With some straightforward algebraic manipulations, $\Delta^P $  can be simplified as follows: 
\begin{align}\label{power gain} 
\Delta^P  
=&\epsilon \sum^{M}_{m=1}\left[  x_m +\left(\frac{1}{M}  \sum^{M}_{i=1}x_i -x_m\right) \right]\\\nonumber &\times \left[  x_m -\left(\frac{1}{M}  \sum^{M}_{i=1}x_i -x_m\right) \right]
\\\nonumber
=&\epsilon\left[  \frac{1}{M}  \sum^{M}_{i=1}x_i  \right]\left[ 2  \sum^{M}_{i=1}x_i -  \sum^{M}_{i=1}x_i   \right]
=\frac{1}{M}  \epsilon\left[ \sum^{M}_{i=1}x_i  \right] ^2,
\end{align}
which is strictly positive and hence leads to the following corollary.
\begin{corollary}
The use of pinching antennas can strictly reduce the overall transmit power, compared to the case with conventional fixed-location antennas. 
\end{corollary}
{\it Remark 3:} We note that \eqref{power gain} also leads to an important insight into the impact of user clustering on the system performance. In particular, \eqref{power gain} shows that the performance gain of pinching antennas over conventional antennas is related to $\left[ \sum^{M}_{i=1}x_i  \right] ^2$. Therefore, the users whose locations are close to each other, e.g., they are distributed on the same side of the service area, should be grouped together, which enhances the performance gain of pinching antennas.

\subsection{Outage Performance Analysis}
Recall that both problems \eqref{pb:1} and \eqref{pb:6} choose the same antenna location, $ x^* =\frac{1}{M}  \sum^{M}_{m=1}x_m$, which ensures that no user is too far away from the pinching antenna. However, even with this optimal location, the users' distances to the pinching antenna are still diverse, which motivates the following outage performance analysis based on the optimal choice $x^*$. 

In particular, based on the use of the optimal power control and antenna location shown in Corollary \ref{corollary1}, we define the following power outage probability:
\begin{align}\nonumber
\mathbb{P}^{\rm o}_m \triangleq& \mathbb{P}\left(
P_m^* \geq \bar{P}
\right)
\\ \label{outageppr}
=& \mathbb{P}\left(
\epsilon (x^*-x_m)^2 +\tau_m \geq \bar{P},
\right)
\end{align}
where $\bar{P}$ denotes the transmit power budget per user. For the considered scenario, the power outage probability is identical to the conventional data rate outage probability. In particular, if $x^*$ is used as the antenna location and $\bar{P}$ is used as the transmit power, ${\rm U}_m$'s data rate is given by $\frac{1}{M}\log\left(
  1+\frac{ \eta \bar{P}}{\sigma^2  \left| {\boldsymbol \psi} _m^* - {\boldsymbol \psi}_{\rm OMA}\right|^2}
  \right) $, and it is straightforward to show that the data rate outage probability, $\mathbb{P}\left(  \frac{1}{M}\log\left(
  1+\frac{ \eta \bar{P}}{\sigma^2  \left| {\boldsymbol \psi} _m^* - {\boldsymbol \psi}_{\rm OMA}\right|^2}
  \right) \leq R
\right)$, is identical to \eqref{outageppr}.

With some straightforward algebraic manipulations,  the power outage probability can be expressed as follows:
\begin{align}
\mathbb{P}^{\rm o}_m =&    \mathbb{P}\left(
 \frac{\left(\sum^{M}_{i=1,i\neq m}x_i-(M-1)x_m\right)^2}{M^2}  +y_m^2 \geq \frac{\bar{P}}{\epsilon}-d^2
\right),
\end{align}
if $\frac{\bar{P}}{\epsilon}>d^2$. Otherwise $\mathbb{P}^{\rm o}_m=1$. Due to the large number of random variables, it is challenging to evaluate the power outage probability in a general multi-user case.  However, for the two-user special case, a closed-form expression for the power outage probability can be obtained as shown in the following lemma. 

\begin{lemma}\label{lemma2}
Assume that the two users are independently and uniformly distributed in the considered service area. The power outage probability achieved by the pinching-antenna system is given by  
\begin{align}\label{eqlemma2}
\mathbb{P}^{\rm o}_m =&    
    \frac{2}{D_{\rm W}}\left(\frac{D_{\rm W}}{2}- \min\left\{\frac{D_{\rm W}}{2}, \sqrt{\frac{\theta_1 D_{\rm L}^2}{4}}\right\}\right)\\\nonumber &+  \frac{2}{D_{\rm W}} \left(g(\theta_2) - g(\theta_3)\right),
\end{align}
if $\frac{\bar{P}}{\epsilon}>d^2$, otherwise $\mathbb{P}^{\rm o}_m=1$, 
where $g(y)$ is defined as follows:
\begin{align}
g(y)=&y+\theta_1y-\frac{4}{3D_{\rm L}^2}y^3  
 \\\nonumber 
&   -\frac{4}{D_{\rm L}}
   \left(
   \frac{y}{2} \sqrt{\theta_4 - y^2} + \frac{\theta_4}{2} \sin^{-1}\left(\frac{y}{\sqrt{\theta_4}}\right)
   \right),
\end{align}
if $\theta_2\geq \theta_3$, otherwise $g(y)=0$,  $\theta_1 =  \frac{4}{D_{\rm L}^2}\left(\frac{\bar{P}}{\epsilon}-d^2\right)$, 
$\theta_2=\min\left\{\frac{D_{\rm W}}{2},\sqrt{\frac{D_{\rm L}^2}{4}\theta_1}\right\}$ and $\theta_3=\sqrt{ \max\left\{0, \frac{D_{\rm L}^2}{4}\left(\theta_1-1\right)\right\}}$, and $\theta_4=\frac{D_{\rm L}^2}{4}\theta_1$. 
\end{lemma}

{\it Remark 4:} Lemma \ref{lemma2} demonstrates how challenging it is to obtain a closed-form expression for the power outage probability, where its expression is quite involved even for the two-user special case. The reason for this challenge is due to the large number of random variables, e.g., $x_m$ and $y_m$, $1\leq m \leq M$.

\section{Greedy-Allocation-Based OMA Transmission}\label{section 4}
The previous section shows that the use of user-fairness-oriented OMA transmission strategies leads to a particular choice of the antenna location which ensures that no user is too far away from the base station, i.e., $x^*=\frac{1}{M}\sum^{M}_{m=1}x_m$.  However, such a choice is no longer optimal if greed-allocation-based OMA transmission strategies are used. In particular, in this section, we consider the following system throughput maximization problem:
 \begin{problem}\label{pb:8} 
  \begin{alignat}{2}
\underset{P_{m}\geq0 , x  }{\rm{max}}  &\quad   \sum^{M}_{m=1}R_m^{\rm OMA}  \label{8tst:1}
\\ s.t. &\quad  R_m^{\rm OMA}\geq R\quad 1\leq m \leq M\\ &\quad \sum^{M}_{m=1}  P_{m}\leq P, \quad -\frac{D_{\rm L}}{2}\leq x \leq \frac{D_{\rm L}}{2}   \label{8tst:2}   .
  \end{alignat}
\end{problem} 
Unlike problems \eqref{pb:1} and \eqref{pb:6}, problem \eqref{pb:8} is much more challenging to solve, mainly due to the non-convex objective function in \eqref{8tst:1}. In order to obtain insightful understandings, the two-user special case, $M=2$, is focused on in the following. 

Similar to the proof for Lemma \ref{lemma1}, the closed-form solutions are first obtained for $P_m$ by assuming that the antenna location, $x$, is fixed. Therefore, problem \eqref{pb:8} is simplified as follows:
 \begin{problem}\label{pb:9} 
  \begin{alignat}{2}
\underset{P_{m}  }{\rm{max}}  &\quad   \sum^{M}_{m=1}\log\left(
  1+\frac{\eta P_m}{\sigma^2  \left| {\boldsymbol \psi} _m - {\boldsymbol \psi}_{\rm OMA}\right|^2}
  \right)   
\\ s.t. &\quad P_m\geq \epsilon (x-x_m)^2+\tau_m  \quad 1\leq m \leq M  \label{9tst:1} \\ &\quad \sum^{M}_{m=1}  P_{m}\leq P  \label{9tst:2}   .
  \end{alignat}
\end{problem}  
It is straightforward to verify that $P$ needs to be no smaller than $\epsilon (\tau_{1x}+\tau_{2x})$, otherwise, problem \eqref{pb:9} is infeasible. Therefore,  $P\geq\epsilon (\tau_{1x}+\tau_{2x})$ is assumed in the remainder of the section. 
 The following lemma provides a closed-form solution of problem \eqref{pb:9}.

 \begin{lemma}\label{lemma3}
Assuming that $x$ is fixed, the optimal solution of problem \eqref{pb:9} is given by
  \begin{align}
\left\{\begin{array}{l}  
 P_1=P  -\epsilon \tau_{2x}  \\ 
 P_2=\epsilon \tau_{2x}  
 \end{array}\right.,
 \end{align} 
 if $
 \frac{1}{P  -\epsilon \tau_{2x} +\frac{\sigma^2 \tau_{1x}}{\eta }  }  - \frac{1}{\epsilon \tau_{2x}+\frac{\sigma^2 \tau_{2x}}{\eta } } \geq 0$. If $
 \frac{1}{P  -\epsilon \tau_{2x} +\frac{\sigma^2 \tau_{1x}}{\eta }  }  - \frac{1}{\epsilon \tau_{2x}+\frac{\sigma^2 \tau_{2x}}{\eta } } <0 $ and $\frac{1}{P  -\epsilon \tau_{1x} +\frac{\sigma^2 \tau_{2x}}{\eta }  }  - \frac{1}{\epsilon \tau_{1x}+\frac{\sigma^2 \tau_{1x}}{\eta } } \geq 0$, the optimal solution is given by
    \begin{align}
\left\{\begin{array}{l}  
 P_1=\epsilon \tau_{1x}  \\
 P_2=P  -\epsilon \tau_{1x}  
 \end{array}\right..
 \end{align} 
 Otherwise, the optimal solution is given by
   \begin{align}\label{kktx34}
\left\{\begin{array}{l}  
P_1=\frac{P}{2}+    \frac{\sigma^2 \tau_{2x}}{2\eta }- \frac{\sigma^2 \tau_{1x}}{2\eta } \\ 
P_2 = \frac{P}{2}+  \frac{\sigma^2 \tau_{1x}}{2\eta }   - \frac{\sigma^2 \tau_{2x}}{2\eta }  
 \end{array}\right..
 \end{align} 
 \end{lemma}
 \begin{proof}
 See Appendix \ref{proof3}. 
 \end{proof}
 {\it Remark 5:} Lemma \ref{lemma3} illustrates how challenging to solve the considered throughput maximization problem, where there are multiple possible expressions for its optimal solution, even for the two-user special case. 
 
By using Lemma \ref{lemma3}, the optimal solution of problem \eqref{pb:9} can be obtained by carrying out a simple one-dimensional search. In particular, for each point on the waveguide, Lemma \ref{lemma3} can be used to find the maximum throughput offered by this point. By enumerating all the points on the waveguide, the optimal solution of problem \eqref{pb:9} can be obtained.   
 
\subsection{A High-SNR Performance Analysis} \label{subsection high snr}
To find an analytical expression for the optimal solution of $x$, a single concise expression is needed for the optimal solution of $P_m$, which is shown to be impossible by Lemma \ref{lemma3}. However, we note that at high SNR, i.e., $P\rightarrow \infty$, $
 \frac{1}{P  -\epsilon \tau_{2x} +\frac{\sigma^2 \tau_{1x}}{\eta }  }  - \frac{1}{\epsilon \tau_{2x}+\frac{\sigma^2 \tau_{2x}}{\eta } } <0 $ and $\frac{1}{P  -\epsilon \tau_{1x} +\frac{\sigma^2 \tau_{2x}}{\eta }  }  - \frac{1}{\epsilon \tau_{1x}+\frac{\sigma^2 \tau_{1x}}{\eta } } < 0$, which leads to the following corollary.
 
 \begin{corollary}\label{corollary3}
 By assuming that $x$ is fixed, at high SNR, the solution shown in \eqref{kktx34} is optimal for problem \eqref{pb:9}. 
 \end{corollary}
 
 By using Lemma \ref{lemma3} and Corollary \ref{corollary3}, problem \eqref{pb:9} can be approximated at high SNR as follows: 
  \begin{problem}\label{pb:10} 
  \begin{alignat}{2}
\underset{  x  }{\rm{max}}  &\quad  \log\left(
1+\frac{\eta \left(\frac{P}{2}+    \frac{\sigma^2 \tau_{2x}}{2\eta }- \frac{\sigma^2 \tau_{1x}}{2\eta }\right)}{\sigma^2 \tau_{1x}}
\right) \\\nonumber & \quad\quad+  \log\left(
1+\frac{\eta \left(\frac{P}{2}+  \frac{\sigma^2 \tau_{1x}}{2\eta }   - \frac{\sigma^2 \tau_{2x}}{2\eta } \right)}{\sigma^2 \tau_{2x}}
\right)  .
  \end{alignat}
\end{problem}  
With some straightforward algebraic manipulations, problem \eqref{pb:10} can be recast to the following equivalent form:  
  \begin{problem}\label{pb:11} 
  \begin{alignat}{2}
\underset{  x  }{\rm{max}}  &\quad  
 \frac{\tau_{1x}}{4\tau_{2x}} + \frac{\tau_{2x}}{4\tau_{1x}}   + 
\frac{\eta P}{2\sigma^2 \tau_{1x}}
\\\nonumber &\quad+ \frac{\eta P}{2\sigma^2\tau_{2x}} +\frac{\eta^2 P^2}{4\sigma^4 \tau_{1x}\tau_{2x}}. 
  \end{alignat}
\end{problem}  
Again by applying the high SNR assumption, problem \eqref{pb:11} can be approximated as follows:
  \begin{problem}\label{pb:12} 
  \begin{alignat}{2}
\underset{  x  }{\rm{min}}  &\quad   f(x)\triangleq \left((x-x_1)^2+y_1^2+d^2\right)\left((x-x_2)^2+y_2^2+d^2\right). 
  \end{alignat}
\end{problem}  

 \begin{figure}[t] \vspace{-0.2em}
\begin{center}
\subfigure[Random case 1 ]{\label{fig2a}\includegraphics[width=0.15\textwidth]{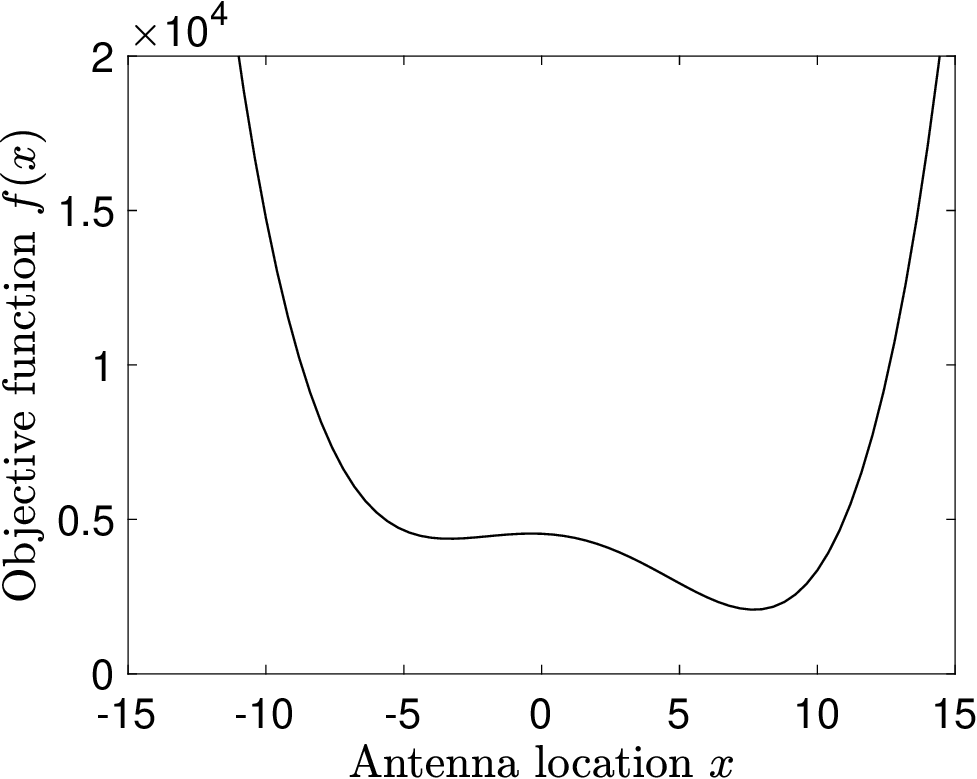}} \subfigure[Random case 2]{\label{fig2b}\includegraphics[width=0.15\textwidth]{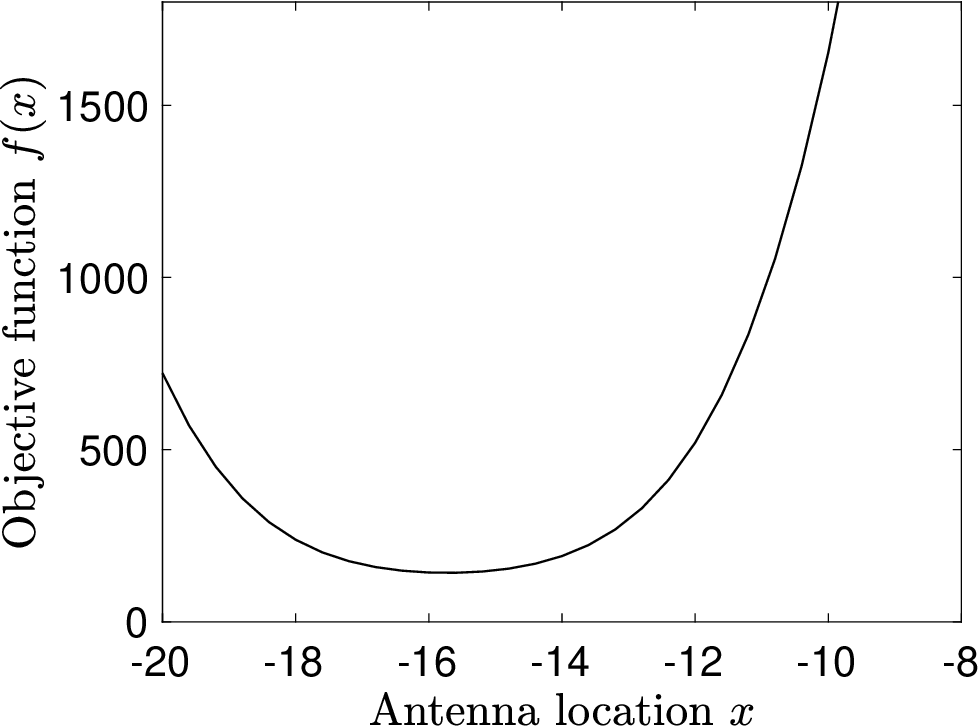}} 
\subfigure[Random case 3]{\label{fig2c}\includegraphics[width=0.15\textwidth]{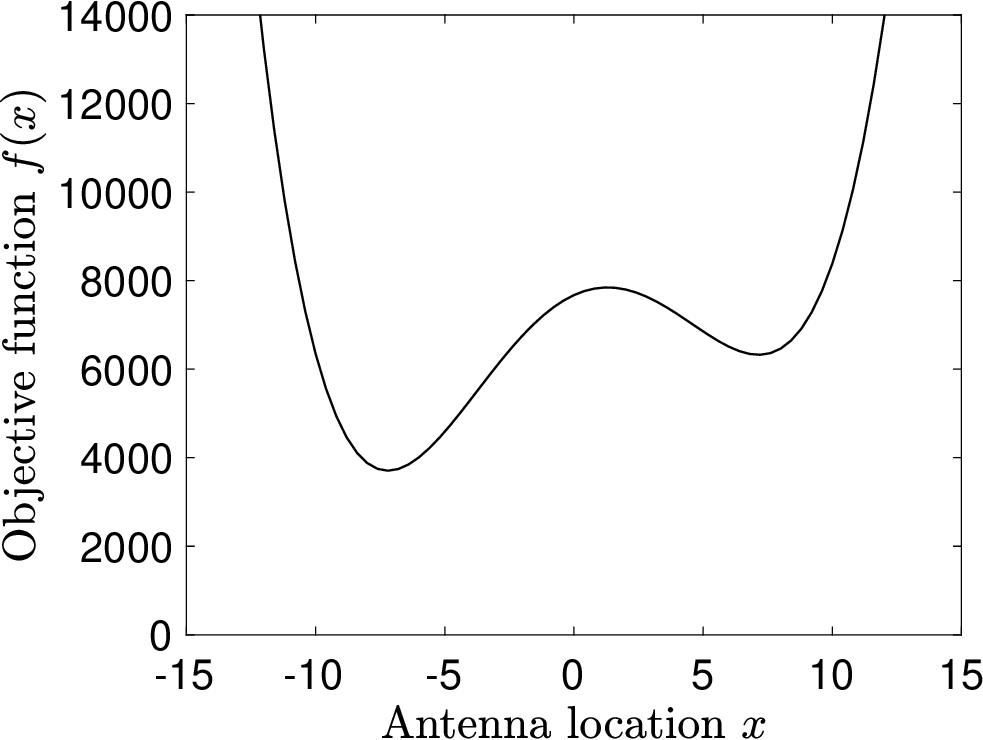}} \vspace{-1.5em}
\end{center}
\caption{The convexity of the objective function of problem \eqref{pb:12}, where $D_{\rm W}=10$ m,  $D_{\rm L}=40$ m  and $d=3$ m. \vspace{-1em} }\label{fig2}\vspace{-1.2em}
\end{figure}

As shown in Fig. \ref{fig2}, $f(x)$ is not a convex function, and may have multiple minima.   
The first order derivative of $f(x)$ is given by
    \begin{align}\label{cubic}
 f'(x)=
 &2(x-x_1)\left((x-x_2)^2+y_2^2+d^2\right)\\\nonumber
 &+2(x-x_2)\left((x-x_1)^2+y_1^2+d^2\right) ,
 \end{align} 
 where the optimal solution of $x$ must be one of the roots of the above cubic function. Therefore, a low-complexity approach to solve problem \eqref{pb:9} is to find all the roots of $f'(x)=0$, calculate the corresponding throughputs achieved by these roots by using Lemma \ref{lemma3}, and then select the best root which yields the largest throughput. Compared to the one-dimensional search-based approach, this root-based approach is of low complexity since there are at most three roots for the cubic function in \eqref{cubic}. 
 
 
 Although a closed-form expression for the optimal antenna location cannot be obtained for problem \eqref{pb:9}, an insightful understanding can be obtained as shown in the following lemma. 
 
 \begin{lemma}\label{lemma4}
If $|y_1|\leq |y_2|$, the pinching antenna should be placed closer to ${\rm U}_1$, i.e., $|x^*-x_1|\leq |x^*-x_2|$. Otherwise, $|x^*-x_1|\geq |x^*-x_2|$. 
 \end{lemma}
 \begin{proof}
 See Appendix \ref{proof4}.
 \end{proof}

\section{Resource Allocation for NOMA Transmission}\label{section 4}
In this section, we consider the situation in which all users are served simultaneously via NOMA. SIC is a key component of NOMA systems, but the implementation of SIC makes the decision of the antenna location more challenging. Take power-domain NOMA as an example, which requires the users to be ordered according to their channel conditions; however, a change of the location of the pinching antenna can change the users' channel conditions.

\subsection{An Intuitive Approach}\label{section 41}
An intuitive approach is to build two resource allocation problems by assuming that   ${\rm U}_1$ and ${\rm U}_2$ carry out SIC, respectively.  For example, by assuming that ${\rm U}_1$ carries out SIC (e.g., ${\rm U}_2$ detects its own signal directly), a resource allocation problem can be built as follows. 

In particular, ${\rm U}_2$'s achievable data rate in NOMA is given by \cite{Nomading}
  \begin{align}\label{models}
  R_2^{\rm NOMA} =&  \log\left(
  1+\frac{\frac{ \eta P_2}{ \left| {\boldsymbol \psi} _2 - {\boldsymbol \psi}_{\rm NOMA}\right|^2}}{\frac{ \eta P_1}{  \left| {\boldsymbol \psi} _2 - {\boldsymbol \psi}_{\rm NOMA}\right|^2}+\sigma^2}
  \right)  \\\nonumber
  =&  \log\left(
  1+\frac{\eta P_2 }{ \eta P_1 +\sigma^2  \left| {\boldsymbol \psi} _2 - {\boldsymbol \psi}_{\rm NOMA}\right|^2}
  \right)  ,
  \end{align}     
  where $P_m$ denotes the transmit power for ${\rm U}_m$'s signal.  ${\rm U}_1$ needs to decode ${\rm U}_2$'s signal with the following data rate,   $ R_0^{\rm NOMA}=\log\left(
  1+\frac{\frac{ \eta P_2}{ \left| {\boldsymbol \psi} _1 - {\boldsymbol \psi}_{\rm NOMA}\right|^2}}{\frac{ \eta P_1}{  \left| {\boldsymbol \psi} _1 - {\boldsymbol \psi}_{\rm NOMA}\right|^2}+\sigma^2}
  \right)  $. If $\min\{  R_2^{\rm NOMA} , R_0^{\rm NOMA}\}\geq R$, SIC can be carried out successfully by ${\rm U}_1$, and ${\rm U}_1$ can decode its own signal with the following data rate: 
     \begin{align}\label{models}
  R_1^{\rm NOMA}  
  =&  \log\left(
  1+\frac{\eta P_1 }{  \sigma^2  \left| {\boldsymbol \psi} _1 - {\boldsymbol \psi}_{\rm NOMA}\right|^2}
  \right). 
  \end{align}

This section is to focus on the total transmit power minimization problem as follows:
     \begin{problem}\label{pb:13} 
  \begin{alignat}{2}
\underset{P_{m}\geq0 , x  }{\rm{min}} &\quad   \sum^{M}_{m=1}  P_{m}\label{13tst:1}
\\ s.t. &\quad    R_m^{\rm NOMA}\geq R, \quad 0\leq m\leq 2   \label{13tst:2}  
\\  &\quad   P_m\geq 0, \quad -\frac{D_{\rm L}}{2}\leq x \leq \frac{D_{\rm L}}{2} \label{13tst:3} . 
  \end{alignat}
\end{problem}  
With some straightforward algebratic manipulations, problem \eqref{pb:13} can be recast as follows:  
    \begin{problem}\label{pb:14} 
  \begin{alignat}{2}
\underset{P_{m}\geq0 , x  }{\rm{min}} &\quad   \sum^{M}_{m=1}  P_{m}\label{14tst:1}
\\ s.t. &\quad      
-P_2+\frac{\tilde{\epsilon} \eta }{\sigma^2}P_1  + \tilde{\epsilon}  (x-x_2)^2 + \tilde{\tau}_2 \leq 0, \label{14tst:2}  \\
&\quad      
-P_2+\frac{\tilde{\epsilon} \eta }{\sigma^2}P_1  + \tilde{\epsilon}  (x-x_1)^2 + \tilde{\tau}_1 \leq 0 \label{14tst:2x}  
  \\ &\quad   -P_1+ \tilde{\epsilon} (x-x_1)^2+\tilde{\tau}_1  \leq 0
\\  &\quad   P_m\geq 0, \quad -\frac{D_{\rm L}}{2}\leq x \leq \frac{D_{\rm L}}{2} \label{14tst:3}  ,
  \end{alignat}
\end{problem}  
where $\tilde{\tau}_m = \tilde{\epsilon}(y_m^2+d^2)$ and $\tilde{\epsilon}=\frac{\sigma^2}{\eta}(e^{R}-1)$. We note that  $\tilde{\tau}_m $ and $\tilde{\epsilon}$ are defined differently to  $ {\tau}_m $ and $ {\epsilon}$ because in NOMA each user can use all the $M$ time slots.

A closed-form optimal solution for this optimization problem is difficult to obtain, even if the trivial constraints, e.g., $P_m\geq 0$ and $-\frac{D_{\rm L}}{2}\leq x \leq \frac{D_{\rm L}}{2} $, are omitted. In particular, there are three potential forms with complicated feasibility conditions for the optimal solution of problem \eqref{pb:13}. We note that for the case with ${\rm U}_2$ carrying out SIC, there are three additional forms, which makes this intuitive approach less insightful and more complex. 

\subsection{A Low-Complexity Optimal Approach}\label{section 42}
Prior to resource allocation, the users are first ordered based on $y_m^2$. In particular, the users are ordered to ensure that $y_1^2\leq y_2^2$. Because this user ordering is not related to the pinching antenna location, it can be carried out without knowing the optimal antenna location.

The proposed low-complexity optimal approach starts with two assumptions. One is that the optimal solution makes ${\rm U}_1$ the strong user, i.e., $ \left| {\boldsymbol \psi} _1 - {\boldsymbol \psi}_{\rm NOMA}\right|\leq  \left| {\boldsymbol \psi} _2 - {\boldsymbol \psi}_{\rm NOMA}\right|$. The other is that the constraints, $P_m\geq 0$ and $-\frac{D_{\rm L}}{2}\leq x \leq \frac{D_{\rm L}}{2} $, can be omitted. At the end of the section, it will be shown that the obtained optimal solution satisfies the two assumptions. With the two assumptions, problem \eqref{pb:14} can be reduced to the following:
  \begin{problem}\label{pb:15} 
  \begin{alignat}{2}
\underset{P_{m} , x  }{\rm{min}} &\quad   \sum^{M}_{m=1}  P_{m}\label{15tst:1}
\\ s.t. &\quad      
-P_2+\frac{\tilde{\epsilon} \eta }{\sigma^2}P_1  + \tilde{\epsilon}  (x-x_2)^2 + \tilde{\tau}_2 \leq 0 \label{15tst:2}  
  \\ &\quad   -P_1+ \tilde{\epsilon} (x-x_1)^2+\tilde{\tau}_1  \leq 0  .\label{15tst:3}  
  \end{alignat}
\end{problem}  
The following lemma shows the optimal solution of problem \eqref{pb:15}.
\begin{lemma}\label{lemma5}
The optimal solution of problem \eqref{pb:15} is given by
 \begin{align}\label{15op}
\left\{\begin{array}{l}  
P_1^* = \tilde{\epsilon} (x^*-x_1)^2+\tilde{\tau}_1\\ 
P_2^* = \frac{\tilde{\epsilon} \eta }{\sigma^2}P_1  + \tilde{\epsilon}  (x^*-x_2)^2 + \tilde{\tau}_2\\  x^*  
 =   \frac{x_2 }{   e^{R}+1} +\frac{  e^{R} x_1  }{   e^{R}+1} .
 \end{array}\right..
 \end{align}
\end{lemma}
\begin{proof}
See Appendix \ref{proof5}.
\end{proof}
By using Lemma \ref{lemma5}, the following lemma about the optimal solution of the original problem in \eqref{pb:13} can be obtained.
\begin{lemma}\label{lemma6}
For $R\geq \frac{1}{2}$, the solution shown in \eqref{15op} is optimal for problem \eqref{pb:13}. 
\end{lemma}
\begin{proof} 
See Appendix \ref{proof6}.
\end{proof} 

 \subsection{Performance Analysis}  
  The closed-form expression for the optimal antenna location facilitates insightful performance analysis, as shown in this subsection. 
  
{\it 1) Where to place the pinching antenna?}   
 For a large $R$, the optimal antenna location shown in \eqref{15op} can be approximated as follows:
\begin{align}
 x^* \approx  & \frac{x_2 }{   e^{R}+1} +  x_1   \rightarrow x_1, 
   \end{align}
which means that if the target data rate is large, it is preferable to activate the pinching antenna very close to ${\rm U}_1$, a user who is near the waveguide. It is worth pointing out that the value of $y_m$ is important to decide the optimal antenna location, which is unlike the OMA case shown in Section \ref{section 2}.  
 
{\it 2) An analytical comparison to OMA:}
The overall transmit power required by OMA-assisted conventional-antenna systems can be straightforwardly obtained as follows: 
\begin{align} 
P^{\rm OMA} =&  \epsilon(x_1^2+y_1^2+d^2) +  \epsilon(x_2^2+y_2^2+d^2) .
\end{align}
By using \eqref{15op}, the overall transmit power achieved by NOMA is given by  
\begin{align}
P^{\rm NOMA} =&   \tilde{\epsilon}  (x^*-x_2)^2 + \tilde{\tau}_2 + \left(\frac{\tilde{\epsilon} \eta }{\sigma^2}+1\right)\left(\tilde{\epsilon} (x^*-x_1)^2+\tilde{\tau}_1\right) .
\end{align}

To facilitate the performance analysis, define $d^{\rm OMA}_m=x_m^2+y_m^2+d^2$, and $d^{\rm NOMA}_m=(x^*-x_m)^2+y_m^2+d^2$. Therefore, $P^{\rm NOMA} $ is simplified as follows:
\begin{align}
P^{\rm NOMA} =&   \tilde{\epsilon}d^{\rm NOMA}_2 + \left(\frac{\tilde{\epsilon} \eta }{\sigma^2}+1\right) \tilde{\epsilon}d^{\rm NOMA}_1 \\\nonumber
=& \frac{\sigma^2}{\eta}(e^{R}-1)\left( d^{\rm NOMA}_2 + e^R d^{\rm NOMA}_1\right).
\end{align} 

Therefore, the difference between the transmit powers required by OMA and NOMA is given by 
\begin{align}
\Delta_P \triangleq & P^{\rm OMA} - P^{\rm NOMA}   
    \\\nonumber
  =& \frac{\sigma^2}{\eta}\left(e^{MR}-1\right)\left(
  d^{\rm OMA}_1+d^{\rm OMA}_2
  \right)\\\nonumber &-\frac{\sigma^2}{\eta}(e^{R}-1)\left( d^{\rm NOMA}_2 + e^R d^{\rm NOMA}_1\right)  .
\end{align}
which can be approximated with large $R$ as follows:
\begin{align}\nonumber
\Delta_P   
    \approx & \left[\left(e^{R} d^{\rm OMA}_2-d^{\rm NOMA}_2\right) +e^{R}\left(d^{\rm OMA}_1 -   d^{\rm NOMA}_1\right)  \right]\\  &\times \frac{\sigma^2}{\eta}\left(e^{R}-1\right).
\end{align}

As discussed previously, in NOMA, the optimal antenna location is closer to ${\rm U}_1$, which means that  $d^{\rm NOMA}_1\leq d^{\rm OMA}_1$ and hence the second term, $\left(d^{\rm OMA}_1 -   d^{\rm NOMA}_1\right) $, is positive. Although it is possible that $d^{\rm NOMA}_2\geq d^{\rm OMA}_2$, $e^{R} d^{\rm OMA}_2-d^{\rm NOMA}_2$ should be still positive, with a large $R$, due to the expoential factor $e^R$, e.g., $e^3\approx 20.1$ and $e^5\approx 148.4$.  Therefore, with a large $R$, $\Delta_P\geq 0$, i.e., the use of NOMA can reduce the energy consumption, compared to OMA.

\section{Numerical Studies}
In this section, computer simulation results will be used to evaluate the performance of the proposed resource allocation algorithms.  For the conducted computer simulations, the noise power is set as $-90$ dBm,  $d=3$ m, $D_{\rm W}=10$ m, and $f_c=28$ GHz, as in \cite{mypa}.   Conventional-antenna systems are used as the benchmarking scheme, i.e., the base station is equipped with a conventional antenna that is fixed at the center of the service area. We note that the resource optimization problem for the conventional-antenna system is a special case of the pinching-antenna one, where the details of the corresponding optimal power allocation are omitted due to space limitations. Without clustering, the users are uniformally distributed in the considered service area, i.e.,  $-\frac{D_{\rm L}}{2}\leq x_m\leq \frac{D_{\rm L}}{2}$ and $-\frac{D_{\rm W}}{2}\leq y_m\leq \frac{D_{\rm W}}{2}$. For the clustering scheme, the users are scheduled to satisfy the following criterion: $-\frac{D_{\rm L}}{4}\leq x_m\leq -\frac{D_{\rm L}}{8}$ and $-\frac{D_{\rm W}}{2}\leq y_m\leq \frac{D_{\rm W}}{2}$.

 \begin{figure}[t] \vspace{-0.2em}
\begin{center}
\subfigure[$M=2$ ]{\label{fig3a}\includegraphics[width=0.45\textwidth]{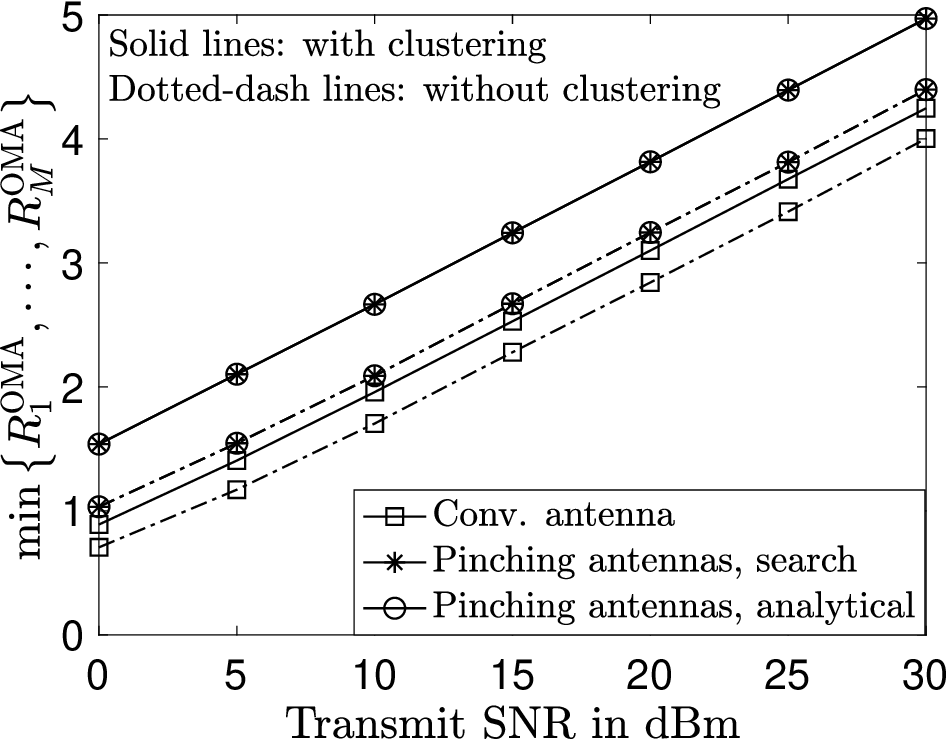}} \subfigure[$M=5$]{\label{fig3b}\includegraphics[width=0.45\textwidth]{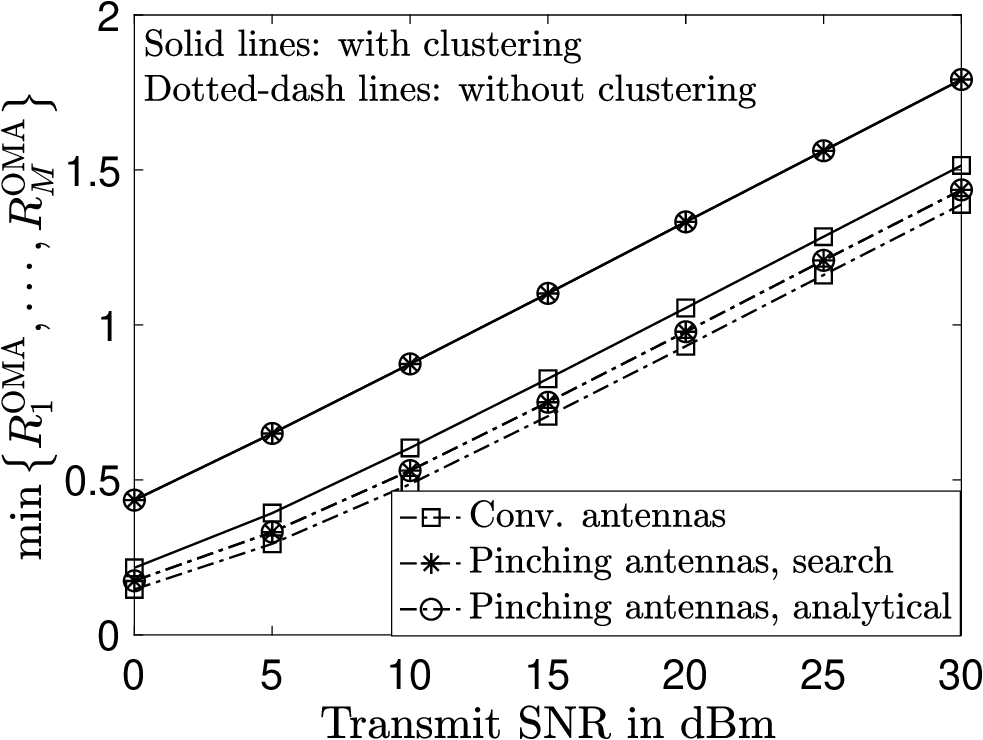}} 
\end{center}
\caption{The user-fairness performance, measured by $ \min\left\{R_1^{\rm OMA}, \cdots, R_M^{\rm OMA}\right\} $ which is the objective function of problem \eqref{pb:1}, achieved by the considered OMA transmission schemes, where   $D_{\rm L}=40$ m.   \vspace{-1em} }\label{fig3}\vspace{-1.2em}
\end{figure}

In Fig. \ref{fig3}, the considered OMA transmission schemes are evaluated by using the user-fairness performance metric,  $ \min\left\{R_1^{\rm OMA}, \cdots, R_M^{\rm OMA}\right\} $ which is the objective function of problem \eqref{pb:1}.  Recall that a closed-form expression for the optimal solution of problem \eqref{pb:1} can be obtained, as shown in Lemma \ref{lemma1}. To evaluate the accuracy of the obtained analytical optimal solution, the performance of an exhaustive-search-based algorithm is also shown in the figure. In particular, a one-dimensional search is carried out to enumerate all possible antenna locations, where problem \eqref{pb:1} becomes a classic convex optimization problem with respect to the power allocation coefficients only.  Fig. \ref{fig3} shows that the analytical solution in Lemma \ref{lemma1} achieves the same performance as the search-based algorithm, which verifies the optimality of the obtained analytical solution. In addition, Fig. \ref{fig3} also demonstrates that the use of clustering can significantly improve the performance gain of pinching antennas over conventional antennas. An interesting observation from Fig. \ref{fig3} is that, without clustering, the performance gap between pinching and conventional antennas is reducing by increasing $M$, which can be explained in the following. As shown in Lemma \ref{lemma1}, the otpimal antenna location is given by $x^*=\frac{1}{M}  \sum^{M}_{m=1}x_m $. If $M\rightarrow \infty$, the optimal location for the pinching antenna becomes the center of the service area, which is identical to conventional-antenna systems. 

 \begin{figure}[t] \vspace{-0.2em}
\begin{center}
\subfigure[$M=2$ ]{\label{fig4a}\includegraphics[width=0.45\textwidth]{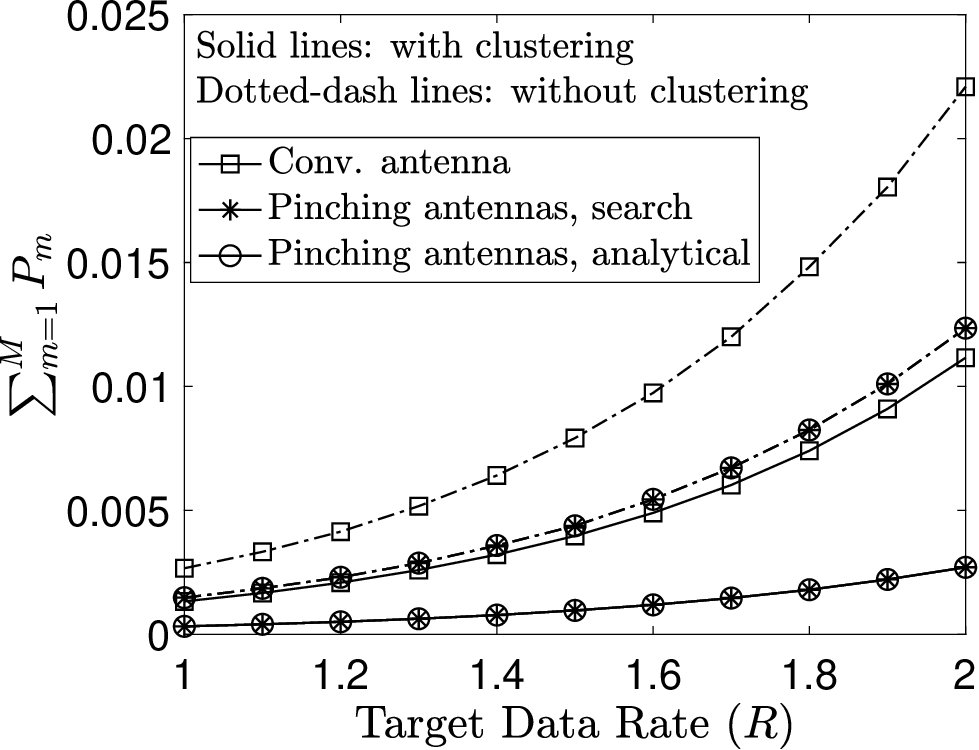}} \subfigure[$M=5$]{\label{fig4b}\includegraphics[width=0.45\textwidth]{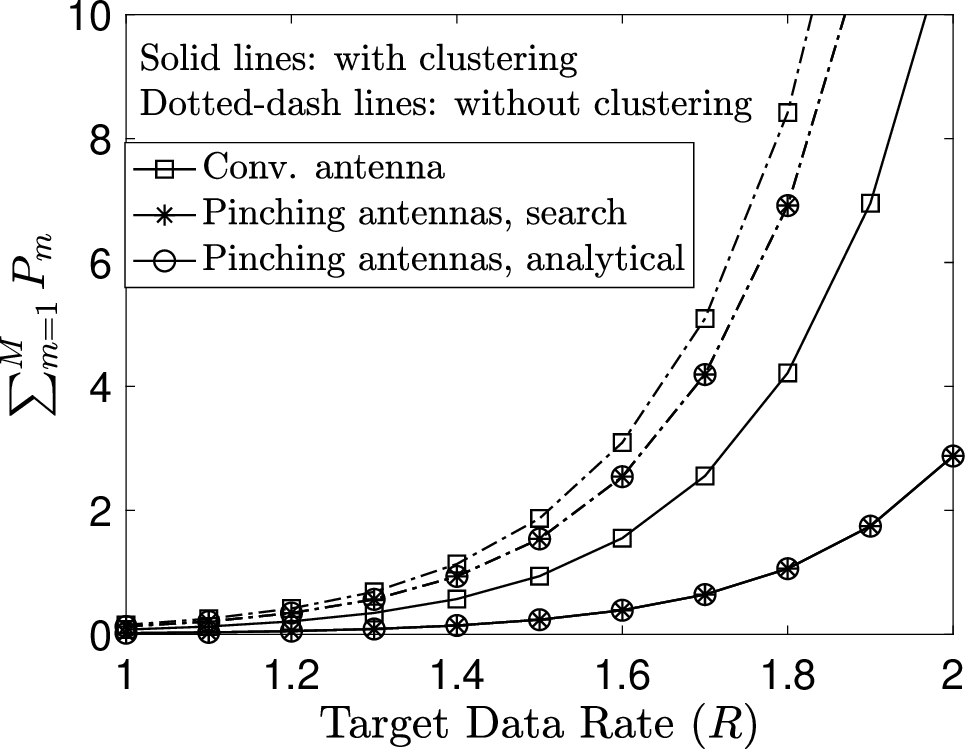}} 
\end{center}
\caption{The total transmit power, i.e., $  \sum^{M}_{m=1}  P_{m}$ which is the objective function of problem \eqref{pb:6}, required by the considered OMA transmission schemes, where   $D_{\rm L}=40$ m.    \vspace{-1em} }\label{fig4}\vspace{-1.2em}
\end{figure}

Fig. \ref{fig4} shows the total transmit power, i.e., $  \sum^{M}_{m=1}  P_{m}$ which is the objective function of problem \eqref{pb:6}, required by the considered OMA transmission schemes. As can be seen from the figure, the use of pinching antennas can significantly reduce the power consumption with or without clustering, for the case of $M=2$. However, for the case of $M=5$, without clustering, the performance achieved by pinching antennas becomes similar to that of conventional antennas, which is again due to the fact that the optimal location of the pinching antenna becomes close to the center of the service area, if $M$ is large. However, with clustering, the performance gain of pinching antennas over conventional antennas again becomes significant, an observation similar to that of Fig. \ref{fig3}. Fig. \ref{fig4} also verifies the optimality of the analytical solution shown in Corollary \ref{corollary1}.  Recall that for the cases considered in Figs. \ref{fig3} and \ref{fig4},  $x^*=\frac{1}{M}  \sum^{M}_{m=1}x_m $ is used as the optimal antenna location, and Fig. \ref{fig5} is provided to show the outage data rates, i.e., $(1-\mathbb{P}^{\rm o}_m)R$, achieved by the considered transmission scheme based on this optimal antenna location. As can be seen from Fig. \ref{fig5}, the use of pinching antennas leads to a significant performance improvement compared to conventional antennas, particularly for the case with clustering.  Fig. \ref{fig5} also demonstrates the accuracy of the analytical results for the outage probability shown in Lemma \ref{lemma2}. 

    \begin{figure}[t]\centering \vspace{-0em}
    \epsfig{file=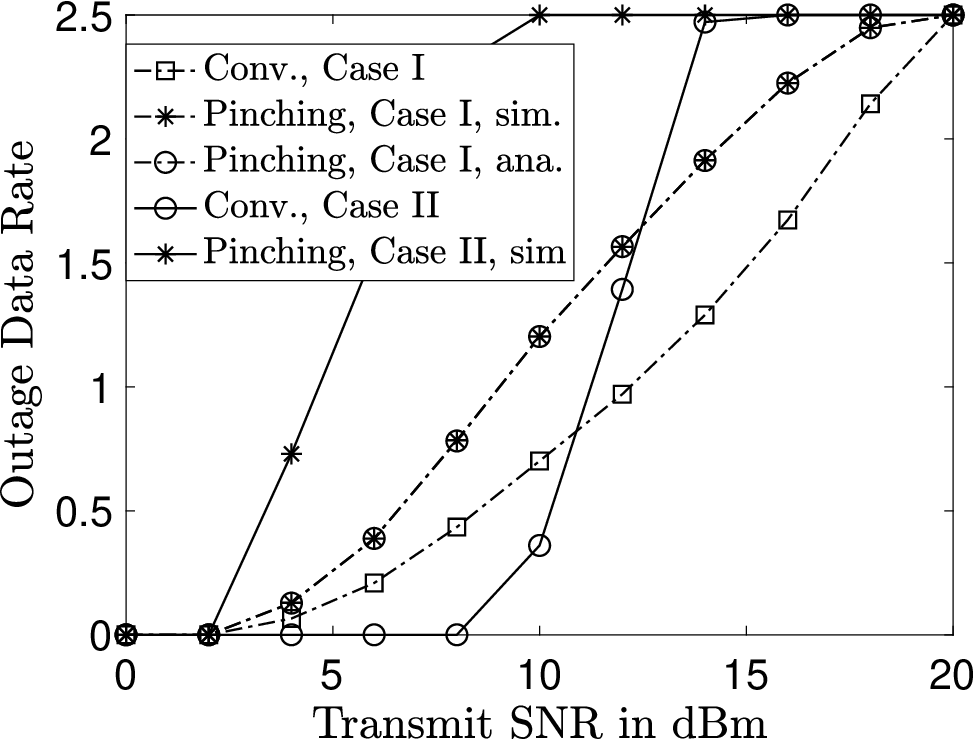, width=0.45\textwidth, clip=}\vspace{-0.5em}
\caption{{ Outage data rates, i.e., $(1-\mathbb{P}^{\rm o}_m)R$, achieved by the considered OMA transmission schemes,  for the case that $x^*=\frac{1}{M}  \sum^{M}_{m=1}x_m $, where $D_{\rm L}=40$ m, and $R=2.5$ bits per channel use (BPCU).   }
  \vspace{-1em}    }\label{fig5}   \vspace{-1em} 
\end{figure}

 \begin{figure}[t] \vspace{-0.2em}
\begin{center}
\subfigure[$R=1$ ]{\label{fig6a}\includegraphics[width=0.45\textwidth]{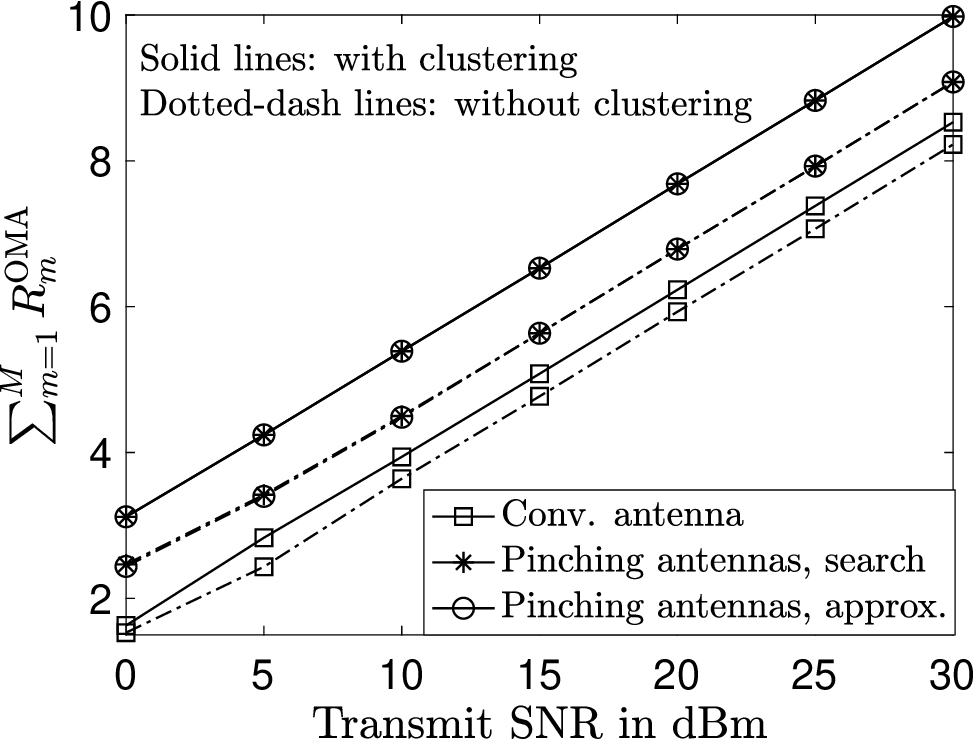}} \subfigure[$R=3$]{\label{fig6b}\includegraphics[width=0.45\textwidth]{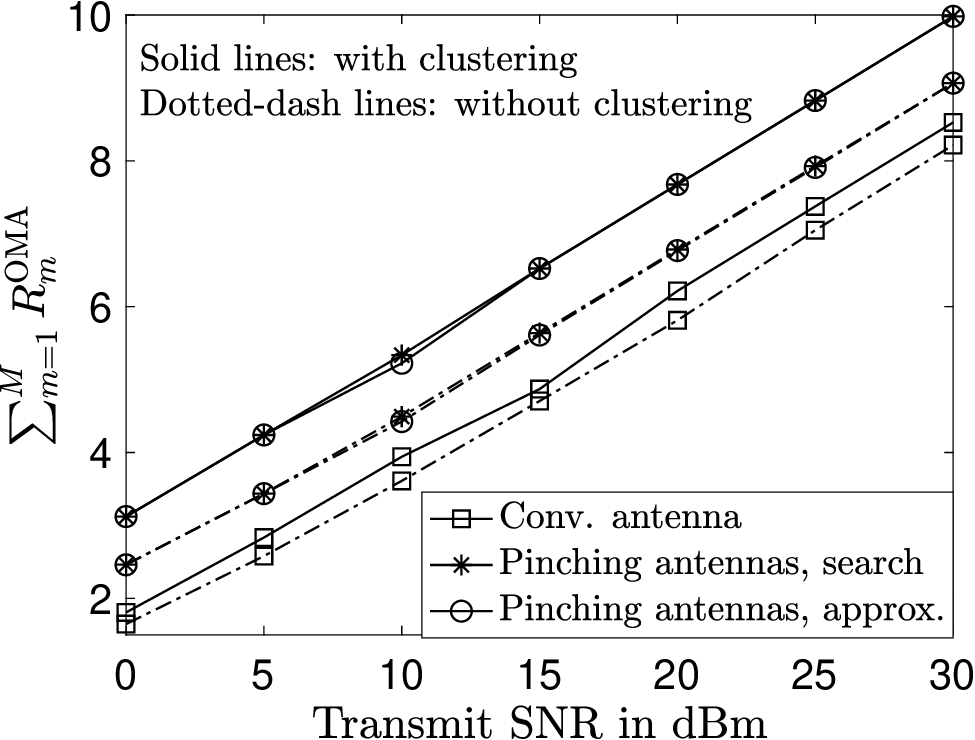}} 
\end{center}
\caption{The system throughput, measured by $ \sum^{M}_{m=1}R_m^{\rm OMA} $ which is the objective function of problem \eqref{pb:8}, achieved by the considered OMA transmission schemes, where  $D_{\rm L}=40$ m and $M=2$. \vspace{-1em} }\label{fig6}\vspace{-1.2em}
\end{figure}

Fig. \ref{fig6} illustrates the system throughput, measured by $ \sum^{M}_{m=1}R_m^{\rm OMA} $ which is the objective function of problem \eqref{pb:8}, achieved by the considered OMA transmission schemes. Recall that, due to its non-convex nature, problem \eqref{pb:8} is more challenging than problems \eqref{pb:1} and \eqref{pb:6}. As discussed in Section \ref{section 4}, a one-dimensional search-based algorithm can be developed by utilizing the solution shown in Lemma \ref{lemma3}. In addition, Section \ref{subsection high snr} also shows a high-SNR approximation for the optimal solution, by applying a root solver to \eqref{cubic}. In Fig. \ref{fig6}, the performance of the high-SNR approximation of the optimal solution is almost identical to that of the search-based scheme, which demonstrates the tightness of the carried-out approximation.  Consistent with the previous figures, Fig. \ref{fig6} demonstrates that the use of clustering is useful to improve the performance of pinching-antenna systems. 

 \begin{figure}[t] \vspace{-0.2em}
\begin{center}
\subfigure[Without clustering ]{\label{fig7a}\includegraphics[width=0.45\textwidth]{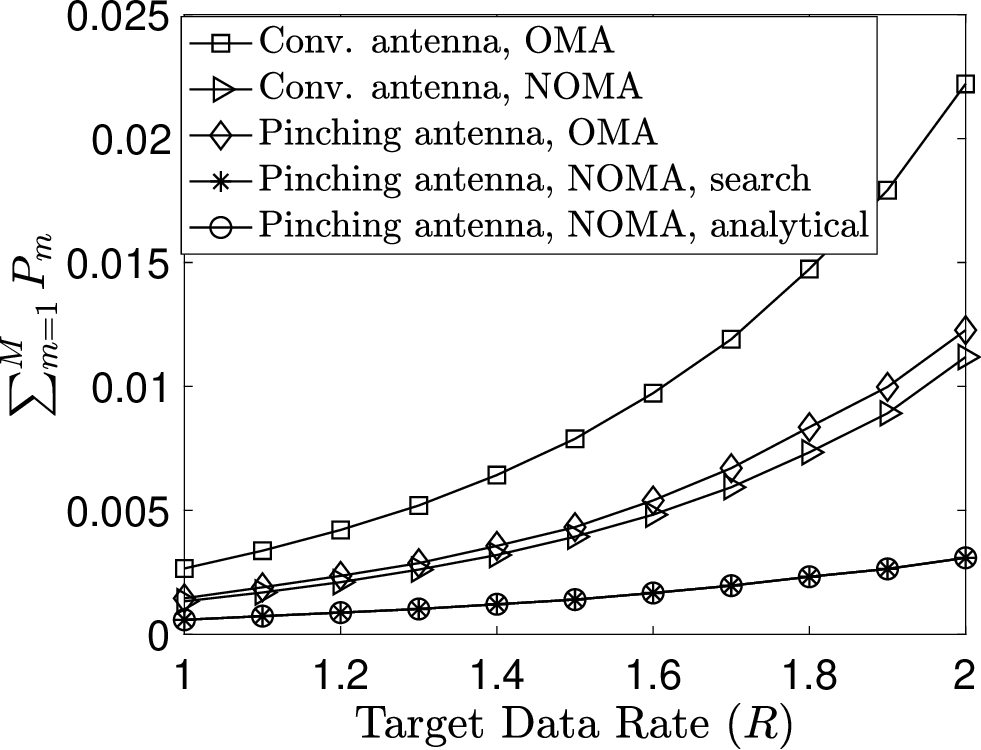}} \subfigure[With clustering]{\label{fig7b}\includegraphics[width=0.45\textwidth]{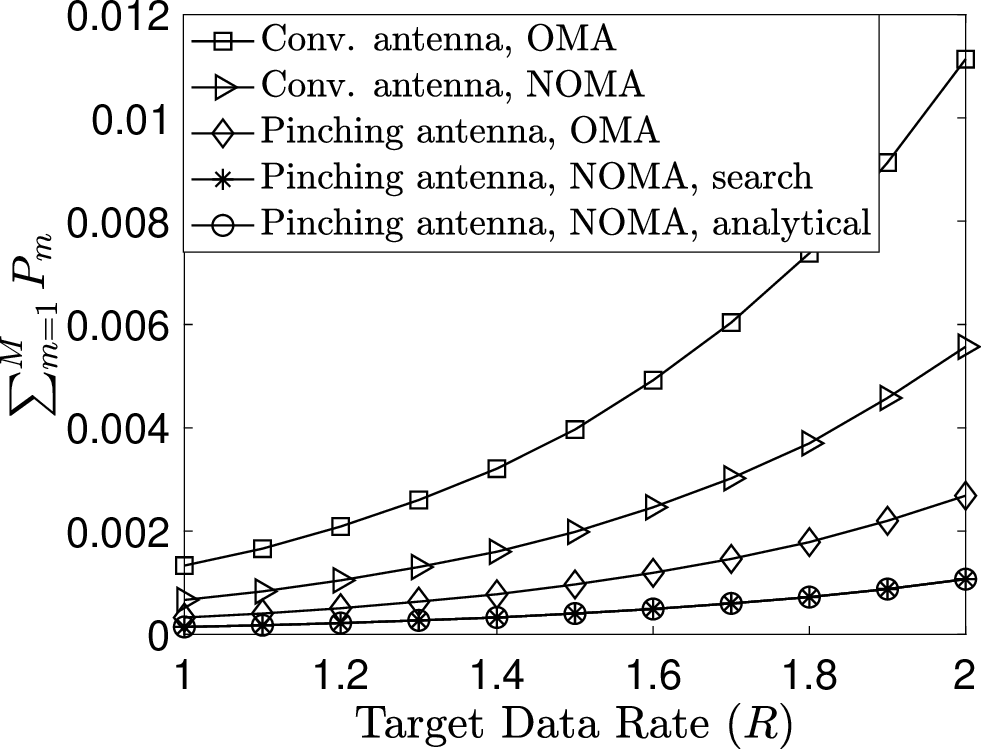}} 
\end{center}
\caption{The total transmit power, i.e., $  \sum^{M}_{m=1}  P_{m}$ which is the objective function of problem \eqref{pb:13}, required by the considered NOMA transmission scheme,  where $D_{\rm L}=40$ m, $M=2$ and $R=2$ BPCU. \vspace{-1em} }\label{fig7}\vspace{-1.2em}
\end{figure}

Fig. \ref{fig7} illustrates the total transmit power, i.e., $  \sum^{M}_{m=1}  P_{m}$ which is the objective function of problem \eqref{pb:13}, required by the considered NOMA transmission scheme. As discussed in Section \ref{section 42}, a closed-form expression for the optimal solution of problem \eqref{pb:13} can be obtained in a low-complexity manner. The optimality of the obtained analytical solution shown in Lemma \ref{lemma5} is verified by the fact that in Fig. \ref{fig7}, the analytical solution achieves the same performance as the search-based scheme. Figs. \ref{fig7a} and \ref{fig7b} show that NOMA-assisted pinching-antenna scheme achieves the best performance among the considered transmission schemes, where the OMA-based conventional-antenna one is the worst. The comparison between the NOMA-assisted conventional-antenna scheme and the OMA-based pinching-antenna system is more complex. Without clustering, they achieve similar performance, with the NOMA-based conventional-antenna scheme being slightly better. With clustering, the OMA-based pinching-antenna scheme significantly outperforms the  NOMA-based conventional-antenna scheme. This significant performance improvement of the OMA-based pinching-antenna system is consistent with the observations previously made in Figs. \ref{fig3} - \ref{fig6}. 

\section{Conclusions}
In this paper, analytical optimization for pinching antenna placement has been carried out in order to reveal the impact of antenna placement on the system performance. In particular, for user-fairness-oriented OMA-based transmission, a closed-form optimal solution was obtained to reveal that the pinching antenna needs to be activated at the place that would be beneficial to all served users; however, the users' distances to the waveguide have no impact on the location selection.  For greedy-allocation-based OMA transmission, an asymptotic study based on the high SNR approximation was carried out to show that the optimal antenna location is in close proximity to the user who is nearest to the waveguide.  For NOMA-based transmission, even with a user-fairness-oriented objective, the analytical results were obtained to show that the optimal antenna location is not the position that can benefit all users, but rather is near the user positioned close to the waveguide.

\appendices
\section{Proof for Lemma \ref{lemma1}}\label{proof1}

By applying the epigraph of the objective function shown in \eqref{1tst:1}, problem \eqref{pb:1} can be recast as follows: 
    \begin{problem}\label{pb:2} 
  \begin{alignat}{2}
\underset{P_{m} \geq0, x,t  }{\rm{max}}  &\quad  t \label{2tst:1}
\\ s.t. &\quad     R_m^{\rm OMA}\geq t, \quad 1\leq m \leq M
\\ &\quad 
\eqref{1tst:2},  \eqref{1tst:3}  . 
  \end{alignat}
\end{problem}  
It is challenging to directly solve problem \eqref{pb:2} which is not a joint convex function of $P_{m}$, $x$, and $t$. In the following, closed-form expressions for the optimal solutions of $P_{m}$ and $t$ are first obtained by fixing $x$. Then, the optimal solution of the antenna location is obtained by utilizing the closed-form solution for $P_{m}$ and $t$. 

In particular, by treating $x$ as a constant, problem \eqref{pb:2} can be recast as follows:
  \begin{problem}\label{pb:3} 
  \begin{alignat}{2}
\underset{P_{m}\geq0 , t }{\rm{max}}  &\quad  t \label{3tst:1}
\\ s.t. &\quad     
    \tau_{mx} \frac{\sigma^2}{\eta}e^{Mt}-\tau_{mx}\frac{\sigma^2}{\eta}-P_m\leq 0 , \quad 1\leq m \leq M \label{3tst:1}
\\ &\quad 
 \sum^{M}_{m=1}  P_{m}\leq P  \label{3tst:2}  ,
   \end{alignat}
\end{problem}  
where $\tau_{mx}=(x-x_m)^2+y_m^2+d^2$ is a constant for a fixed $x$. It is straightforward to verify that the constraints shown in \eqref{3tst:1} are convex, and hence problem \eqref{pb:3} is a convex optimization problem. 

The Lagrangian of problem \eqref{pb:3}  is given by
\begin{align}
L =& -t +\sum^{M}_{m=1}\lambda_m\left( \tau_{mx}\frac{\sigma^2}{\eta}e^{Mt}-\tau_{mx}\frac{\sigma^2}{\eta}-P_m
\right) \\\nonumber &+\lambda_0\left( \sum^{M}_{m=1}  P_{m}- P \right),
\end{align}
where $\lambda_i$, $0\leq i \leq M$, denote the Lagrange multipliers. We note that the constraints $P_m\geq 0$ are omitted here in order to reduce the number of Lagrange multipliers, where whether the obtained power control solutions are non-negative will be verified later. 
The KKT conditions of problem \eqref{pb:3} are given by
\begin{align}
\left\{\begin{array}{l}  
-1 +M \frac{\sigma^2}{\eta}e^{Mt} \sum^{M}_{m=1}\lambda_m\tau_{mx} =0
 \\ 
- \lambda_m  +\lambda_0=0, 1\leq m \leq M\\ 
\lambda_m\left( \tau_{mx}\frac{\sigma^2}{\eta}e^{Mt}-\tau_{mx}\frac{\sigma^2}{\eta}-P_m\right) =0 , 1\leq m \leq M\\
\lambda_0\left(
 \sum^{M}_{m=1}  P_{m}- P
\right)=0
 \end{array}\right..
 \end{align}
 The solutions of the above KKT conditions are given by
\begin{align}\label{kktsol1}
\left\{\begin{array}{l}  
\lambda_0 = \lambda_m = \frac{\eta}{M \sigma^2e^{Mt} \sum^{M}_{m=1}\tau_{mx}}
 \\ 
  t=\frac{1}{M}\log \frac{\eta P+\sigma^2 \sum^{M}_{m=1} \tau_{mx} }{\sigma^2 \sum^{M}_{m=1}   
\tau_{mx} }\\P_m =  \frac{ \tau_{mx}   }{  \sum^{M}_{i=1}   \tau_{ix} } P
 \end{array}\right..
 \end{align}
 Because both $P_m$, $1\leq m \leq M$, and $\lambda_i$, $0\leq i\leq M$, shown in \eqref{kktsol1} are positive, the obtained solutions are the optimal solutions of problem \eqref{pb:3}.  Due to space limitations, the primal feasibility conditions are omitted here, and it is straightforward to verify that the obtained solution satisfies all the primal feasibility conditions.

By substituting the closed-form expressions of $P_m$ and $t$ in \eqref{kktsol1} into problem \eqref{pb:2}, an optimization problem with respect to $x$ only is obtained as follows:
    \begin{problem}\label{pb:4} 
  \begin{alignat}{2}
\underset{  -\frac{D_{\rm L}}{2}\leq x \leq \frac{D_{\rm L}}{2}  }{\rm{max}}  &\quad  \frac{1}{M}\log \frac{\eta P+\sigma^2 \sum^{M}_{m=1} \tau_{mx} }{\sigma^2 \sum^{M}_{m=1}   
\tau_{mx} } \label{4tst:1} ,
  \end{alignat}
\end{problem}  
which is equivalent to the following minimization problem:  
    \begin{problem}\label{pb:5} 
  \begin{alignat}{2}
\underset{ -\frac{D_{\rm L}}{2}\leq x \leq \frac{D_{\rm L}}{2}  }{\rm{min}}  &\quad   \sum^{M}_{m=1}   \left(
(x-x_m)^2+y_m^2+d^2\right)\label{5tst:1} .
  \end{alignat}
\end{problem}  
Problem \eqref{pb:5} is a convex optimization problem, and its Lagrangian is given by $
L =  \sum^{M}_{m=1}   \left(
(x-x_m)^2+y_m^2+d^2\right) + \lambda_1\left(x - \frac{D_{\rm L}}{2}  \right)+ \lambda_2\left(- \frac{D_{\rm L}}{2} -x \right)$, 
where $\lambda_i$, $1\leq i \leq 2$, denote the Lagrange multipliers. The KKT conditions of problem \eqref{pb:5} are given by
\begin{align}
\left\{\begin{array}{l}  
2\sum^{M}_{m=1}   \left(
x-x_m\right) + \lambda_1 - \lambda_2 =0
 \\ 
\lambda_1\left(x - \frac{D_{\rm L}}{2}  \right)=0\\
\lambda_2\left(- \frac{D_{\rm L}}{2} -x \right)=0 \end{array}\right..
 \end{align}
 We note that the solution obtained later satisfies the primal feasibility conditions which are omitted in the above equation set. 
If  $\lambda_1\neq 0$, $x=\frac{D_{\rm L}}{2} $, $\lambda_2=0$, and   $\lambda_1=-2\sum^{M}_{m=1}   \left(
\frac{D_{\rm L}}{2} -x_m\right) $. Because $-\frac{D_{\rm L}}{2}\leq x_m\leq \frac{D_{\rm L}}{2}$, $\lambda_1<0$, which means that $\lambda_1\neq 0$ cannot be true. Following similar steps, it can be established that $\lambda_2=0$, which means that the optimal solution needs to satisfy the following constraint:
\begin{align}
 &2\sum^{M}_{m=1}  
(x^*-x_m)=0,
\end{align}
which completes the proof for the lemma.

\section{Proof for Lemma \ref{lemma2}} \label{proof2}
Without loss of generality, ${\rm U}_1$'s outage probability is focused on in the proof. For the two-user special case, the user's power outage probability can be expressed as follows: 
\begin{align}
\mathbb{P}^{\rm o}_1 =&    \mathbb{P}\left(
 \frac{1}{4} \left( x_2- x_1\right)^2 +y_1^2 \geq \frac{\bar{P}}{\epsilon}-d^2
\right).
\end{align}
Recall that $x_m$ is unformlly distributed between $-\frac{D_{\rm L}}{2}$ and $\frac{D_{\rm L}}{2}$, instead of a standard uniform distribution between $0$ and $1$. To facilitate the performance analysis, $\mathbb{P}^{\rm o}_1$ is rewritten as follows:
\begin{align}
\mathbb{P}^{\rm o}_1  
=&    \mathbb{P}\left(
 \frac{D_{\rm L}^2}{4}z_4 +y_1^2 \geq \frac{\bar{P}}{\epsilon}-d^2
\right),
\end{align}
where $z_4=z_3^2$ $z_3=z_1-z_2$, $z_1=\frac{1}{D_{\rm L}}x_1+\frac{1}{2}$, and $z_2=\frac{1}{D_{\rm L}}x_2 +\frac{1}{2}$.

Because $x_1$ and $x_2$ are uniformly distributed between $-\frac{D_{\rm L}}{2}$ and $\frac{D_{\rm L}}{2}$, both $z_1$ and $z_2$ are uniformly distributed between $0$ and $1$. Therefore, the probability density function (pdf) of $z_3$ follows the triangular distribution, i.e., 
\begin{align}
f_{z_3}(z) = 1 - |z|, 
\end{align}
for $-1\leq z \leq 1$. By using $f_{z_3}(z) $ and also the fact that $z_4=z_3^2$, the pdf of $z_4$ is obtained as follows:
\begin{align}
f_{z_4}(z) =& f_{z_3}(\sqrt{z})  \left|\frac{d}{dz} \sqrt{z} \right| + f_{z_3}(-\sqrt{z})   \left|\frac{d}{dz} (-\sqrt{z}) \right|
\\\nonumber =&
  2   (1 - \sqrt{z})   \frac{1}{2\sqrt{z}} =\frac{1}{\sqrt{z}}-1.
\end{align} 
Therefore, the cumulative distribution function (CDF) of $z_4$ is given by
\begin{align}
  F_{z_4}(z) = 2\sqrt{z}-z,
\end{align}
for $0< z \leq 1$. 

Therefore, the power outage probability can be expressed as follows:
\begin{align}
\mathbb{P}^{\rm o}_1 =&      \mathbb{P}\left(
z_4  \geq  \frac{4}{D_{\rm L}^2}\left(\frac{\bar{P}}{\epsilon}-d^2\right)-\frac{4}{D_{\rm L}^2}y_1^2
\right). 
\end{align}
Because $z_4$ is confined between $0$ an $1$,  $\mathbb{P}^{\rm o}_1$ needs to be expressed as follows:
\begin{align}
\mathbb{P}^{\rm o}_1 =&     \mathbb{P}\left(
 \theta_1-\frac{4}{D_{\rm L}^2}y_1^2\leq 0
\right)\\\nonumber &+  \mathbb{P}\left(
z_4  \geq \theta_1-\frac{4}{D_{\rm L}^2}y_1^2 , 0<\theta_1-\frac{4}{D_{\rm L}^2}y_1^2\leq 1
\right).
\end{align}

Recall that there is a hiddent constraint to $y_1$, i.e., $-\frac{D_{\rm W}}{2}\leq y_1\leq \frac{D_{\rm W}}{2}$, which means that the power outage probability can be further expressed as follows:
\begin{align}
\mathbb{P}^{\rm o}_1 =&    
   2 \mathbb{P}\left(
 y_1\geq \min\left\{\frac{D_{\rm W}}{2}, \sqrt{\frac{\theta_1 D_{\rm L}^2}{4}}\right\}
\right)+\\\nonumber & 2\mathbb{P}\left(
1\geq z_4  \geq \theta_1-\frac{4}{D_{\rm L}^2}y_1^2 , \min\left\{\frac{D_{\rm W}}{2},\sqrt{\frac{D_{\rm L}^2}{4}\theta_1}\right\}>y_1\right. \\\nonumber &\left.\quad\quad\quad\geq\sqrt{ \max\left\{0,\frac{D_{\rm L}^2}{4}\left(\theta_1-1\right)\right\}}
\right).
\end{align}
By applying the CDF of $z_4$,  $\mathbb{P}^{\rm o}_1$ can be rewritten as follows:
\begin{align}
\mathbb{P}^{\rm o}_1 =&   
   2  \frac{\frac{D_{\rm W}}{2}- \min\left\{\frac{D_{\rm W}}{2}, \sqrt{\frac{\theta_1 D_{\rm L}^2}{4}}\right\}}{D_{\rm W}}\\\nonumber &+  2\frac{1}{D_{\rm W}} \int^{\theta_2}_{\theta_3}\left(1-2\sqrt{\theta_1-\frac{4}{D_{\rm L}^2}y_1^2}+\theta_1-\frac{4}{D_{\rm L}^2}y_1^2  \right)dy_1  \\\nonumber   =&   
   2  \frac{\frac{D_{\rm W}}{2}- \min\left\{\frac{D_{\rm W}}{2}, \sqrt{\frac{\theta_1 D_{\rm L}^2}{4}}\right\}}{D_{\rm W}}\\\nonumber &+  2\frac{1}{D_{\rm W}}  \left[y_1+\theta_1y_1-\frac{4}{3D_{\rm L}^2}y_1^3  
\right.\\\nonumber 
&\left.   -\frac{4}{D_{\rm L}}
   \left(
   \frac{y_1}{2} \sqrt{\theta_4 - y_1^2} + \frac{\theta_4}{2} \sin^{-1}\left(\frac{y_1}{\sqrt{\theta_4}}\right)
   \right) \right]^{\theta_2}_{\theta_3} .
\end{align}
With some straightforward algebraic manipulations, the expression shown in \eqref{eqlemma2} can be obtained,
which completes the proof of the lemma. 

\section{Proof for Lemma \ref{lemma3}}\label{proof3}

It is straightforward to verify that problem \eqref{pb:9} is a convex optimization problem of $P_m$, and therefore its optimal solution can be obtained by solving its KKT conditions. In particular, the Lagrangian of problem \eqref{pb:9} is given by
\begin{align}
L =& -\sum^{M}_{m=1}\log\left(
  1+\frac{\eta P_m}{\sigma^2 \tau_{mx}}
  \right)  \\\nonumber
  &+ \sum^{M}_{m=1} \lambda_m\left( \epsilon (x-x_m)^2+\tau_m-P_m\right)\\\nonumber
  &+\lambda_0\left(
  \sum^{M}_{m=1}  P_{m}- P
  \right),
\end{align}
where $\lambda_m$, $0\leq m \leq M$, denote the Lagrange multipliers. 
The corresponding KKT conditions are given by
\begin{align}
\left\{\begin{array}{l}   -
\frac{
   \frac{ \eta}{\sigma^2 \tau_{mx}}}{
  1+\frac{\eta P_m}{\sigma^2 \tau_{mx}}}- \lambda_m +\lambda_0=0, 1\leq m \leq M\\ 
 \lambda_m\left( \epsilon (x-x_m)^2+\tau_m-P_m\right)=0, 1\leq m \leq M\\
\lambda_0\left(
 \sum^{M}_{m=1}  P_{m}- P
\right)=0\\ \eqref{9tst:1},\eqref{9tst:2}
 \end{array}\right..
 \end{align}
 It is straightforward to show that  $\lambda_0\neq 0$, otherwise $ -
\frac{
   \frac{ \eta}{\sigma^2 \tau_{mx}}}{
  1+\frac{\eta P_m}{\sigma^2 \tau_{mx}}}- \lambda_m=0$, which is not true. Therefore, the KKT conditions can be rewritten as follows:  
  \begin{align}\label{kktx1}
\left\{\begin{array}{l}  
P_m = \frac{1}{\lambda_0 -\lambda_m}- \frac{\sigma^2 \tau_{mx}}{\eta }, 1\leq m \leq M\\ 
 \lambda_m\left( \epsilon (x-x_m)^2+\tau_m-P_m\right)=0, 1\leq m \leq M\\
 \sum^{M}_{m=1}  P_{m}- P
 =0\\ \eqref{9tst:1},\eqref{9tst:2}
 \end{array}\right..
 \end{align}
 The KKT conditions can be solved by further analyzing the possible choices of $\lambda_1$ and $\lambda_2$.

 \subsubsection{The case with $\lambda_1\neq 0$ and $\lambda_2\neq 0$} if both the multipliers are non-zero, the KKT conditions in \eqref{kktx1} are reduced to the following: 
  \begin{align}
\left\{\begin{array}{l}  
P_m = \frac{1}{\lambda_0 }- \frac{\sigma^2 \tau_{mx}}{\eta }, 1\leq m \leq M\\  
 \sum^{M}_{m=1}  P_{m}- P
 =0
 \end{array}\right.,
 \end{align}
 which is not feasible 
 given the random locations of the users. 

 \subsubsection{The case with $\lambda_1=0$ and $\lambda_2\neq 0$} with these choices of the multipliers,  the KKT conditions in \eqref{kktx1} are reduced as follows: 
  \begin{align}
\left\{\begin{array}{l}  
P_1= \frac{1}{\lambda_0 -\lambda_1}- \frac{\sigma^2 \tau_{1x}}{\eta } \\ 
P_2 = \frac{1}{\lambda_0 -\lambda_2}- \frac{\sigma^2 \tau_{2x}}{\eta } \\ 
 \epsilon (x-x_2)^2+\tau_2-P_2=0 \\
 \sum^{M}_{m=1}  P_{m}- P
 =0\\ \eqref{9tst:1},\eqref{9tst:2}
 \end{array}\right..
 \end{align}
 With some straightforward algebraic manipulations, the closed-form expressions can be obtained for $P_m$ and $\lambda_m$ as follows: 
   \begin{align}\label{kktx133}
\left\{\begin{array}{l}  
 P_1=P  -\epsilon \tau_{2x}  \\ 
 P_2=\epsilon \tau_{2x}  \\
 \lambda_0  =\frac{1}{P+\frac{\sigma^2 \tau_{1x}}{\eta } -\epsilon \tau_{2x}}\\
\lambda_2=\frac{1}{P  -\epsilon \tau_{2x} +\frac{\sigma^2 \tau_{1x}}{\eta }  }  - \frac{1}{\epsilon \tau_{2x}+\frac{\sigma^2 \tau_{2x}}{\eta } } 
 \end{array}\right..
 \end{align} 
 Due to the assumption that $P\geq \epsilon(\tau_{1x}+\tau_{2x})$, $\lambda_0$ is always non-negative. Among the two primary feasibility conditions, \eqref{9tst:2} is always satisfied. By substituting the solutions in \eqref{kktx133} into the primary feasibility condition, \eqref{9tst:1}, the following inequality is obtained:
  \begin{align} 
 \epsilon \tau_{1x}- P  +\epsilon \tau_{2x}\leq 0,
 \end{align}
 which always holds due to the assumption that $P\geq \epsilon (\tau_{1x}+\tau_{2x})$. 
    Therefore, the solution corresponding to this case is optimal if $\lambda_2\geq 0$.

  \subsubsection{The case with $\lambda_1\neq 0$ and $\lambda_2= 0$} following the analysis for the case with $\lambda_1=0$ and $\lambda_2\neq 0$, it is straightforward to establish that the solutions for this case are given by
   \begin{align}\label{kktx2}
\left\{\begin{array}{l}  
 P_1=\epsilon \tau_{1x}  \\
 P_2=P  -\epsilon \tau_{1x}  \\ 
 \lambda_0  =\frac{1}{P+\frac{\sigma^2 \tau_{2x}}{\eta } -\epsilon \tau_{1x}}\\
\lambda_1=\frac{1}{P  -\epsilon \tau_{1x} +\frac{\sigma^2 \tau_{2x}}{\eta }  }  - \frac{1}{\epsilon \tau_{1x}+\frac{\sigma^2 \tau_{1x}}{\eta } } 
 \end{array}\right.,
 \end{align} 
 where its feasibility condition can be obtained as shown in the lemma. 
  
 \subsubsection{The case with $\lambda_1=0$ and $\lambda_2= 0$} with such choices of the multipliers, the KKT conditions in \eqref{kktx1} can be used to obtain the following closed-form solutions:
  \begin{align}\label{kktx3}
\left\{\begin{array}{l}  
P_1=\frac{P}{2}+    \frac{\sigma^2 \tau_{2x}}{2\eta }- \frac{\sigma^2 \tau_{1x}}{2\eta } \\ 
P_2 = \frac{P}{2}+  \frac{\sigma^2 \tau_{1x}}{2\eta }   - \frac{\sigma^2 \tau_{2x}}{2\eta } \\
 \lambda_0 =\frac{2}{P+  \frac{\sigma^2 \tau_{1x}}{\eta }  + \frac{\sigma^2 \tau_{2x}}{\eta }}
 \end{array}\right..
 \end{align} 
  By combining \eqref{kktx1}, \eqref{kktx2}, and \eqref{kktx3}, the proof of the lemma is complete. 

\section{Proof for Lemma \ref{lemma4}}\label{proof4}

Only the case with $|y_1|\leq |y_2|$ is focused on in the proof. Without loss of generality, assume $x_1\leq x_2$.  The lemma is proved by the following two steps. The first step is to show that there is at least one local minimum between $x_1$ and the middle point $\frac{x_1+x_2}{2}$. By using the first-order derivative expression in \eqref{cubic}, the following equalities can be established:
     \begin{align}\nonumber
 f'(x_1)=& 2(x_1-x_2)\left((x-x_1)^2+y_1^2+d^2\right) \leq 0,\\\nonumber
  f'(x_2)=
 &2(x_2-x_1)\left((x-x_2)^2+y_2^2+d^2\right)\geq 0,
 \end{align}
 which are due to the assumption that $x_1\leq x_2$.
 In addition, $ f\left(\frac{x_1+x_2}{2}\right) $ can be expressed as follows:
      \begin{align} 
  f\left(\frac{x_1+x_2}{2}\right) =& (x_1-x_2)\left( y_1^2-y^2_2 \right) \geq 0,
 \end{align}
 which follows from the assumptions that  $|y_1|<|y_2|$ and $x_1\leq x_2$. Therefore, there is at least one root between $x_1$ and $\frac{x_1+x_2}{2}$.

The second step is to show that if there exists a local minimum between $\frac{x_1+x_2}{2}$ and $x_2$, denoted by $\tilde{x}$, there exists a point between $x_1$ and $\frac{x_1+x_2}{2}$ which can reduce the objective function. Recall that the objective function with $\tilde{x}$ is given by
    \begin{align}
 f(\tilde{x}) 
 =&\left((\tilde{x}-x_1)^2+y_1^2+d^2\right)\left((\tilde{x}-x_2)^2+y_2^2+d^2\right) .
 \end{align}
 Define $\bar{x}=\tilde{x}-\frac{x_1+x_2}{2}$, and $ f(\tilde{x}) $ can be expressed as follows:
     \begin{align}
 f(\tilde{x}) = \bar{f}(\bar{x})  \triangleq &\left((\bar{x} +\phi_1)^2+\phi_2\right)\left((\bar{x}- \phi_1)^2+c\phi_2\right),
 \end{align}
 where $\phi_1=\frac{x_2-x_1}{2}$, $\phi_2=y_1^2+d^2$,  $c=\frac{y_2^2+d^2}{y_1^2+d^2}$ and $c\geq 1$ due to the assumption that $|y_1|<|y_2|$. 

Consider a new point $\breve{x}=-\bar{x}+\frac{x_1+x_2}{2}$, which is between $x_1$ and the middle point $\frac{x_1+x_2}{2}$, i.e., $\breve{x}-x_1\leq \frac{x_1+x_2}{2}$ since  $x_2-\tilde{x}\leq \frac{x_1+x_2}{2}$. The second step of the proof is complete if the following inequality holds:
\begin{align}
  f(\breve{x}) 
  -f(\tilde{x})\leq 0 .
\end{align}
The above difference can be expressed as follows:
     \begin{align}\nonumber
  f(\breve{x}) 
  -f(\tilde{x}) =&\left((-\bar{x} +\phi_1)^2+\phi_2\right)\left((-\bar{x}- \phi_1)^2+c\phi_2\right)\\\nonumber &-\left((\bar{x} +\phi_1)^2+\phi_2\right)\left((\bar{x}- \phi)^2+c\phi_2\right) \\\nonumber 
=& \phi_2(1-c) \left((\bar{x}+ \phi)^2-(\bar{x} -\phi_1)^2\right)\leq 0,
 \end{align}
 where the last inequality follows from the assumptions that $x_1\leq x_2$ and $|y_1|\leq |y_2|$.  Combining the two aforementioned steps, the proof of the lemma is complete. 
 
 \section{Proof for Lemma \ref{lemma5}}\label{proof5}
 
The Lagrangian of problem \eqref{pb:15} is given by
\begin{align}
L = & \sum^{M}_{m=1}  P_{m}+\lambda_2\left(-P_2+\frac{\tilde{\epsilon} \eta }{\sigma^2}P_1  + \tilde{\epsilon}  (x-x_2)^2 + \tilde{\tau}_2 \right)\\\nonumber &+\lambda_1\left(
  -P_1+ \tilde{\epsilon} (x-x_1)^2+\tilde{\tau}_1
\right).
\end{align}
The corresponding KKT conditions are given by
\begin{align}\label{kktc2}
\left\{\begin{array}{l}  
1-\lambda_2=0\\
1+\lambda_2\frac{\tilde{\epsilon} \eta }{\sigma^2}-\lambda_1=0\\
 2 \lambda_2\tilde{\epsilon}(x-x_2)+2 \lambda_1\tilde{\epsilon}(x-x_1)=0\\ 
\lambda_2\left(-P_2+\frac{\tilde{\epsilon} \eta }{\sigma^2}P_1  + \tilde{\epsilon}  (x-x_2)^2 + \tilde{\tau}_2 \right) =0 \\
\lambda_1\left(
  -P_1+ \tilde{\epsilon} (x-x_1)^2+\tilde{\tau}_1
\right)=0\\\eqref{15tst:2}  \& \eqref{15tst:3}  
 \end{array}\right..
 \end{align}
 The first KKT condition, $1-\lambda_2=0$, leads to the conclusion that $\lambda_2=1$, which yields the following choice of $ \lambda_1$:
 \begin{align}
 \lambda_1= 1+ \frac{\tilde{\epsilon} \eta }{\sigma^2}. 
 \end{align}
 By using the choices of the multipliers and also the KKT conditions, the optimal solutions for the transmit powers are given by  
 \begin{align}\label{p*}
\left\{\begin{array}{l}  
P_1^* = \tilde{\epsilon} (x^*-x_1)^2+\tilde{\tau}_1\\ 
P_2^* = \frac{\tilde{\epsilon} \eta }{\sigma^2}P_1  + \tilde{\epsilon}  (x^*-x_2)^2 + \tilde{\tau}_2.
 \end{array}\right..
 \end{align}
 
By using the third condition in \eqref{kktc2} and also the choices of $\lambda_1$ and $\lambda_2$, the optimal antenna location needs to satisfy $ (x-x_2)+\lambda_1 (x-x_1)=0$, which yields a closed-form expression of the optimal antenna location as follows:
\begin{align} 
 x^* =  \frac{x_2+ x_1(1+ \frac{\tilde{\epsilon} \eta }{\sigma^2})}{ \frac{\tilde{\epsilon} \eta }{\sigma^2}+2} .
   \end{align}
With some straightforward algebraic manipulations, the optimal solution shown in the lemma can be obtained, and the proof of the lemma is complete. 
 
 \section{Proof for Lemma \ref{lemma6}} \label{proof6}
 
In order to show that the obtained solution in \eqref{15op} is also optimal for the original problem in \eqref{pb:13}, it is sufficient to show whether the obtained solution satisfies the two assumptions.

Recall that the first assumption is that the optimal solution makes ${\rm U}_1$ the strong user. Therefore, showing that the first assumption is satisfied is equivalent to proving the following inequality:
\begin{align}
\Delta_d\triangleq (x^*-x_1)^2+y^2_1- (x^*-x_2)^2+y^2_2\leq 0.
\end{align}
 By using the closed-form expression of $x^*$ in \eqref{15op}, $\Delta_d$ can be expressed as follows:
 \begin{align}
\Delta_d\triangleq& \left( \frac{x_2 }{   e^{R}+1} +\frac{  e^{R} x_1  }{   e^{R}+1}-x_1\right)^2+y^2_1\\\nonumber &- \left( \frac{x_2 }{   e^{R}+1} +\frac{  e^{R} x_1  }{   e^{R}+1}-x_2\right)^2-y^2_2 \\\nonumber
=& \frac{ ( x_2- x_1)^2  }{   \left( e^{R}+1\right)^2}\left(1-e^{2R}\right)+y^2_1-y^2_2  \leq 0,
\end{align}
since the users are orderd to ensure $y_1^2\leq y_2^2$ and it is assumed that $R\geq \frac{1}{2}$.  

Recall that the second assumption is that the constraints, $P_m\geq 0$ and $-\frac{D_{\rm L}}{2}\leq x \leq \frac{D_{\rm L}}{2} $, can be satisifed by the obtained solution.  $P_m^*\geq 0$ holds obviously, and $-\frac{D_{\rm L}}{2}\leq x^* \leq \frac{D_{\rm L}}{2} $ can be proved by showing that $x^*$ is between $x_1$ and $x_2$. First, focusing on the case that $x_1\leq x_2$, $x_1\leq x^*\leq x_2$ holds since 
   \begin{align}
 x^* -x_2=   & \frac{- e^{R} x_2 }{   e^{R}+1} +\frac{  e^{R} x_1  }{   e^{R}+1}  \leq 0,
 \\\nonumber
  x^* -x_1=   &  \frac{x_2 }{   e^{R}+1} -\frac{  x_1  }{   e^{R}+1} \geq 0.
   \end{align}
For the case that $x_2\leq x_1$, the fact that $x_2\leq x^*\leq x_1$ can also be proved by following the steps similar to the case of  $x_1\leq x_2$. The proof for the lemma is complete. 
 
  \vspace{-0.5em}
\bibliographystyle{IEEEtran}
\bibliography{IEEEfull,trasfer}
  \end{document}